\documentclass[runningheads]{llncs}

\usepackage[english]{babel}
\usepackage[utf8]{inputenc}
\usepackage{mathtools}
\usepackage{amsfonts, amsmath, amssymb}
\usepackage{float}
\usepackage{hyperref}
\usepackage{tikz}
\usepackage[ruled,lined,linesnumbered]{algorithm2e}

\usepackage[T1]{fontenc}
\usepackage{graphicx}
\usepackage{color}

\renewcommand{\UrlFont}{\color{blue}\rmfamily}

\definecolor{darkgreen}{rgb}{0,0.5,0}
\definecolor{orange}{rgb}{1,0.733,0}
\definecolor{darkorange}{rgb}{0.8,0.4,0}
\definecolor{gray}{rgb}{0.157,0.157,0.157}

\newcommand{\lessspace}{\vspace{-6pt}\setlength{\topsep}{0pt}\setlength{\parskip}{0pt}\setlength{\itemsep}{0pt}}

\let\O\relax
\DeclareMathOperator{\O}{O}

\let\originalleft\left
\let\originalright\right
\renewcommand{\left}{\mathopen{}\mathclose\bgroup\originalleft}
\renewcommand{\right}{\aftergroup\egroup\originalright}

\allowdisplaybreaks[1]

\newcommand{\appendixref}[1]{\hyperref[#1]{Appendix~\ref*{#1}}}
\addto\extrasenglish{}
\addto\extrasenglish{}
\addto\extrasenglish{}
\addto\extrasenglish{}
\addto\extrasenglish{}

\newcommand{\backlink}[2]{\hyperref[#1]{\UrlFont #2 Extended version of \autoref*{#1} (page \pageref*{#1}).}}
\newenvironment{restate}[2]{
	\def\localcountername{#1}%
	\def\localrefname{#2}%
	\noindent\backlink{#2}{$\downarrow$}%
	\expandafter\def\csname the#1\endcsname{\ref*{#2}}%
}{
	\addtocounter{\localcountername}{-1}%
	\noindent\backlink{\localrefname}{$\uparrow$}%
}

\newcommand{\bisimilar}{\mathrel{\sim}}
\newcommand{\wbisimilar}{\mathrel{\approx}}

\begin{document}

\title{Encoding Petri Nets into CCS\thanks{This is an extended version of the paper with the same title published at COORDINATION 2024.}}
\subtitle{Technical Report}

\titlerunning{Encoding Petri Nets into CCS (Technical Report)}

\author{Benjamin Bogø\inst{1}\orcidID{0009-0000-2192-3291} \and Andrea Burattin \inst{1}\orcidID{0000-0002-0837-0183} \and Alceste Scalas \inst{1}\orcidID{0000-0002-1153-6164}}

\authorrunning{B.~Bogø, A.~Burattin, and A.~Scalas}

\institute{Technical University of Denmark, Denmark\\
\email{\{bbog,andbur,alcsc\}@dtu.dk}}

\maketitle

\begin{abstract}
This paper explores the problem of determining which classes of Petri nets can be encoded into behaviourally-equivalent CCS processes. Most of the existing related literature focuses on the inverse problem (i.e., encoding process calculi belonging to the CCS family into Petri nets), or extends CCS with Petri net-like multi-synchronisation (Multi-CCS). In this work, our main focus are \emph{free-choice} and \emph{workflow} nets (which are widely used in process mining to describe system interactions) and our target is \emph{plain} CCS. We present several novel encodings, including one from free-choice workflow nets (produced by process mining algorithms like the $\alpha$-miner) into CCS processes, and we prove that our encodings produce CCS processes that are weakly bisimilar to the original net. Besides contributing new expressiveness results, our encodings open a door towards bringing analysis and verification techniques from the realm of process calculi into the realm of process mining.

\end{abstract}

\section{Introduction}\label{sec:introduction}

Process calculi and Petri nets are among the most successful tools for the modelling and verification of concurrent systems. The two models have significantly different designs: Petri nets have a more \emph{semantic} flavour, whereas process calculi have a more \emph{syntactic} flavour. This has resulted in significantly different approaches and application fields. In particular, Petri nets have found considerable success in the area of \emph{Workflow Management}, as the theoretical foundation for several \emph{Business Process Management} languages, and in \emph{process mining}, whereas the syntactic nature of process calculi has fostered a rich literature on the static verification of behavioural properties (e.g. via type checking or the axiomatisation of bisimulation relations), often connected to programming languages.

This different focus on semantics-vs.-syntax has naturally encouraged the study of Petri nets as a possible semantic model for process calculi, through the development of various encodings and results of the form: \emph{Petri nets (of the class $X$) are at least as expressive as the encoded calculus $Y$}. (For more details, see \autoref{sec:related-work}.)
In this paper we investigate the opposite problem: \emph{Which flavour of Petri nets can be encoded in Milner's Calculus of Communicating Systems (CCS)?} A reason for this investigation is the observation that applications of Petri nets in process mining (e.g.~via the $\alpha$-miner algorithm~\cite{alpha-miner}) often result in rather structured nets (in particular, \emph{free-choice workflow nets}~\cite{free-choice-nets,alpha-miner}) which are reminiscent of what is expressible in CCS. Therefore, we aim at proving whether this intuition is correct. Moreover, besides producing novel expressiveness results, developing an encoding from (selected classes of) Petri nets into CCS could also open new doors towards directly using process calculi in process mining, or applying analysis and verification techniques and tools originally developed for process calculi (e.g.~model checkers) to the realm of process mining.

\paragraph{Contributions and structure.}
\autoref{fig:overview} gives an overview of the relation and conversion between Petri nets classes and CCS considered in this paper. We start by presenting related work in \autoref{sec:related-work} and preliminaries in \autoref{sec:preliminaries}. Then, we present an encoding of \emph{free-choice workflow nets} into weakly bismilar CCS processes (\autoref{thm:free-wf-full}) in \autoref{sec:encoding}. Here, we also introduce a new class of Petri nets called \emph{group-choice nets} (which include free-choice nets) and show how to encode them into weakly bisimilar CCS processes (\autoref{thm:group}). We conclude and outline future work in \autoref{sec:conclusion-future}.
A software tool \cite{pn2ccs} has been created based on the results of this paper: a web application to import/draw Petri nets, classify them (according to \autoref{fig:overview}) and encode them into CCS. Proofs for lemmas and theorems are available in the appendices.

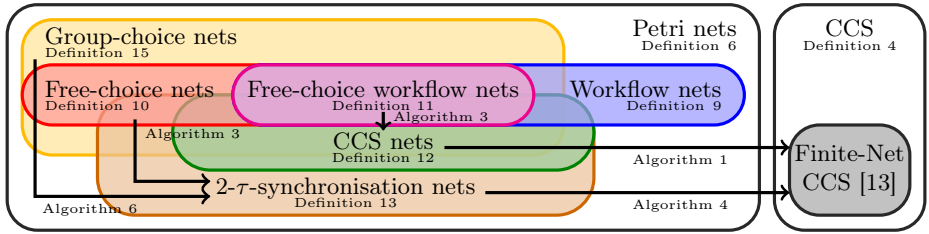
\begin{figure}[!t]
    \centering
    \begin{tikzpicture}[very thick]
        \filldraw[rounded corners=4mm, color=gray, fill=white!40, fill opacity=0.7](0,0) rectangle (10,3);
        \filldraw[rounded corners=4mm, color=gray, fill=white!40, fill opacity=0.7](10.2,0) rectangle (12.2,3);

        \filldraw[rounded corners=4mm, color=orange, fill=orange!40, fill opacity=0.7](0.2,1) rectangle (7.4,2.8);
        \filldraw[rounded corners=5mm, color=darkorange, fill=darkorange!40, fill opacity=0.7](1.2,0.2) rectangle (7.8,1.8);
        \filldraw[rounded corners=5mm, color=darkgreen, fill=darkgreen!40, fill opacity=0.7](2.2,0.8) rectangle (7.8,1.8);

        \filldraw[rounded corners=4mm, color=red, fill=red!40, fill opacity=0.7](0.2,1.4) rectangle (7,2.2);
        \filldraw[rounded corners=4mm, color=blue, fill=blue!40, fill opacity=0.7](3,1.4) rectangle (9.8,2.2);
        \filldraw[rounded corners=4mm, color=magenta, fill=magenta!40, fill opacity=0.7](3,1.4) rectangle (7,2.2);

        \filldraw[rounded corners=4mm, color=gray, fill=gray!40, fill opacity=0.7](10.4,1.4) rectangle (12,0.2);

        \node[align=left, text width=30mm] (gcn) at (2,2.5) {\hyperref[def:group-choice-net]{Group-choice nets\\[-6pt]\tiny{\autoref{def:group-choice-net}}}};
        \node[align=right, text width=24mm] (pn) at (8.5,2.6) {\hyperref[def:petri-net]{Petri nets}\\[-6pt]\tiny{\autoref{def:petri-net}}};
        \node[align=center] (c) at (11.2,2.6) {\hyperref[def:ccs-syntax]{CCS}\\[-6pt]\tiny{\autoref{def:ccs-syntax}}};
        \node[align=left, text width=24mm] (fcn) at (1.7,1.8) {\hyperref[def:free-choice-net]{Free-choice nets}\\[-6pt]\tiny{\autoref{def:free-choice-net}}};
        \node[align=center] (fcwn) at (5,1.8) {\hyperref[def:free-choice-workflow-net]{Free-choice workflow nets}\\[-6pt]\tiny{\autoref{def:free-choice-workflow-net}}};
        \node[align=right, text width=24mm] (wn) at (8.3,1.8) {\hyperref[def:workflow-net]{Workflow nets}\\[-6pt]\tiny{\autoref{def:workflow-net}}};
        \node[align=center] (cn) at (5,1.1) {\hyperref[def:ccs-net]{CCS nets}\\[-6pt]\tiny{\autoref{def:ccs-net}}};
        \node[align=center] (2tsn) at (4.5,0.5) {\hyperref[def:2-tau-synchronisation-net]{2-$\tau$-synchronisation nets}\\[-6pt]\tiny{\autoref{def:2-tau-synchronisation-net}}};
        \node[align=center, text width=20mm] (ccs) at (11.2,0.8) {Finite-Net CCS~\cite{lts-book}};
        \coordinate[] (2tsnt) at (2.7,0.65) {};
        \coordinate[] (2tsnb) at (2.7,0.45) {};
        \coordinate[] (ccs11) at (7.5,1.1) {};
        \coordinate[] (ccs12) at (10.4,1.1) {};
        \coordinate[] (ccs21) at (7.5,0.5) {};
        \coordinate[] (ccs22) at (10.4,0.5) {};

        \draw [->] (fcwn.south) ++(0,0.1) -- node [midway,right] {\tiny\autoref{alg:free-wf-full}} (cn.north) -- ++(0,-0.1);
        \draw [->] (fcn.south) |- node [near start,above right] {\tiny\autoref{alg:free-wf-full}} (2tsnt);
        \draw [->] (gcn.south west) ++(0,+0.1) |- node [midway,below right=-1pt] {\tiny\autoref{alg:group-full}} (2tsnb);
        \draw [->] (cn.east) -- (ccs11) -- node [midway,below=-1pt] {\tiny\autoref{alg:ccs-net}} (ccs12.west);
        \draw [->] (2tsn.east) -- (ccs21) -- node [midway,below=-1pt] {\tiny\autoref{alg:free}} (ccs22.west);
    \end{tikzpicture}
    \caption{Overview of the relation between Petri net classes and CCS considered in this paper. The arrows show algorithms for converting one class into another class.}
    \label{fig:overview}
\end{figure}

\section{Related Work}\label{sec:related-work}

Many process mining algorithms (like the $\alpha$-miner~\cite{alpha-miner}) take a log of traces of \emph{visible} actions and turn it into a \emph{workflow net}~\cite{process-mining-book}. Workflow nets are Petri nets used to describe how systems interact during a process. These kinds of interactions can also be described by process calculi like CCS with labeled semantics to capture the visible actions. There has been some debate about whether graphs like Petri nets or process calculi like the $\pi$-calculus are best for process mining~\cite{petri-net-vs-ccs}. Most of the existing work builds on Petri nets~\cite{process-mining-survey}, but there is also work on how to represent patterns in process calculi~\cite{workflow-patterns-in-ccs}.

Most existing work about encodings between process calculi and Petri nets focus on encoding the former into the latter: e.g.,~there are encodings from variants of CCS~\cite{DBLP:conf/litp/Goltz90,async-ccs-to-open-pn,ccs-with-replication-to-pn}, CSP~\cite{csp-to-labeled-pn}, and finite-control $\pi$-calculus \cite{safe-petri-net-to-ccs} to various classes of Petri nets. \cite{safe-petri-net-to-ccs} also briefly describes an encoding from \emph{unlabeled} safe (1-bounded) Petri nets into CCS with reduction semantics; the result of the encoding is claimed weakly bisimilar to the original net. However, applications in process mining require \emph{labelled} semantics.

To our knowledge, encodings of Petri nets into process calculi are less explored. Gorrieri and Versari~\cite{multi-ccs} present an extension of CCS called \emph{Multi-CCS}, with the purpose of having a one-to-one correspondence between unsafe P/T Petri nets and Multi-CCS; crucially, Multi-CCS can synchronize multiple processes at a time (like Petri nets) whereas CCS is limited to two synchronizing processes at a time. \cite{multi-ccs} also presents an encoding from Petri nets into strongly bisimilar Multi-CCS processes; they also show that encoding a restricted class of Petri nets (called \emph{CCS nets}) yields strongly bisimilar \emph{plain} CCS process. In this paper we start from this last result, and explore encodability beyond CCS nets --- targeting \emph{plain} CCS only (to enable reusing its well-established techniques e.g.~for model checking and axiomatic reasoning), and with an eye torward classes of Petri nets relevant for process mining.

\section{Preliminaries: LTSs, Bisimulations, CCS, Petri Nets}\label{sec:preliminaries}

This section contains the basic standard definitions used in the rest of the paper.

\paragraph{LTSs and bisimulations.} We adopt standard definitions of strong and weak bisimulation between LTS states (\autoref{def:lts}, \ref{def:strong-bisimulation}, and \ref{def:weak-bisimulation}, based on \cite{lts-book}).

\begin{definition}[Labeled transition system]
\label{def:lts}
    A \emph{labeled transition system (LTS)} is a triple $(Q, A, \xrightarrow{\cdot})$ where $Q$ is a set of \emph{states}, $A$ is a set of \emph{actions}, and $\mathop{\xrightarrow{\cdot}} \subseteq Q \times A \times Q$ is a \emph{labelled transition relation}. The set $A$ may contain a distinguished \emph{internal action} $\tau$, and we dub any other action as \emph{visible}. We write:
    \begin{itemize}
        \lessspace
        \item $q \xrightarrow{\mu} q'$ \,iff\, $(q, \mu, q') \in \mathop{\xrightarrow{\cdot}}$
        \item $q \xrightarrow{a} q'$ \,iff\, $a \neq \tau$ and $(q, a, q') \in \mathop{\xrightarrow{\cdot}}$ \hfill(note that the action $a$ is not silent)
        \item $q \xRightarrow{\epsilon} q'$ \,iff\, $q = q'$ or $q \xrightarrow{\tau}\cdots\xrightarrow{\tau} q'$ \hfill(i.e., $q$ can reach $q'$ in 0 or more $\tau$-steps)%
        \item $\smash{q \xRightarrow{a} q'}$ \,iff\, $\smash{q \xRightarrow{\epsilon}\xrightarrow{a}\xRightarrow{\epsilon} q'}$
        \hfill($q$ can reach $q'$ via one $a$-step + 0 or more $\tau$-steps)%
        \lessspace
    \end{itemize}
    We say that $q$ is a \emph{deadlock} if there are no transitions from $q$. A \emph{divergent path} is an infinite sequence of LTS states $q_1,q_2,\ldots$ such that $q_i \xrightarrow{\tau} q_{i+1}$.
\end{definition}

\begin{definition}[Strong bisimulation]
\label{def:strong-bisimulation}
    A \emph{strong bisimulation} between two LTSs $(Q_1, A_1, \xrightarrow{\cdot}_1)$ and $(Q_2, A_2, \xrightarrow{\cdot}_2)$ is a relation $\mathcal{R} \subseteq Q_1 \times Q_2$ where, if $(q_1, q_2) \in \mathcal{R}$:

    \smallskip\centerline{$\begin{array}{c}%
        \forall q_1': q_1 \mathrel{\xrightarrow{\mu}_1} q_1' \;\;\text{implies}\;\; \exists q_2': q_2 \mathrel{\xrightarrow{\mu}_2} q_2' \;\text{and}\; (q_1', q_2') \in \mathcal{R}
        \\
        \forall q_2': q_2 \mathrel{\xrightarrow{\mu}_2} q_2' \;\;\text{implies}\;\; \exists q_1': q_1 \mathrel{\xrightarrow{\mu}_1} q_1' \;\text{and}\; (q_1', q_2') \in \mathcal{R}
    \end{array}$}\smallskip

    \noindent%
    We say that $q$ and $q'$ are \emph{strongly bisimilar} or simply \emph{bisimilar} (written $q \bisimilar q'$) if there exists a bisimulation $\mathcal{R}$ with $(q, q') \in \mathcal{R}$.
\end{definition}

\begin{definition}[Weak bisimulation]
\label{def:weak-bisimulation}
    A \emph{weak bisimulation} between two LTSs $(Q_1, A_1, \xrightarrow{\cdot}_1)$ and $(Q_2, A_2, \xrightarrow{\cdot}_2)$ is a relation $\mathcal{R} \subseteq Q_1 \times Q_2$ where, if $(q_1, q_2) \in \mathcal{R}$:

    \smallskip\centerline{$\begin{array}{c}%
        \forall q_1': q_1 \mathrel{\xrightarrow{a}_1} q_1' \;\;\text{implies}\;\; \exists q_2': q_2 \mathrel{{\xRightarrow{a}}{}_2} q_2' \;\text{and}\; (q_1', q_2') \in \mathcal{R}
        \\
        \forall q_1': q_1 \mathrel{\xrightarrow{\tau}_1} q_1' \;\;\text{implies}\;\; \exists q_2': q_2 \mathrel{{\xRightarrow{\epsilon}}{}_2} q_2' \;\text{and}\; (q_1', q_2') \in \mathcal{R}
        \\
        \forall q_2': q_2 \mathrel{\xrightarrow{a}_2} q_2' \;\;\text{implies}\;\; \exists q_1': q_1 \mathrel{{\xRightarrow{a}}{}_1} q_1' \;\text{and}\; (q_1', q_2') \in \mathcal{R}
        \\
        \forall q_2': q_2 \mathrel{\xrightarrow{\tau}_2} q_2' \;\;\text{implies}\;\; \exists q_1': q_1 \mathrel{{\xRightarrow{\epsilon}}{}_1} q_1' \;\text{and}\; (q_1', q_2') \in \mathcal{R}
    \end{array}$}\smallskip

    \noindent%
    We say that $q$ and $q'$ are \emph{weakly bisimilar} (written $q \wbisimilar q'$) if there is a weak bisimulation $\mathcal{R}$ with $(q, q') \in \mathcal{R}$.
\end{definition}

\paragraph{CCS.} We adopt a standard version of CCS with LTS semantics, including restrictions and defining equations (\autoref{def:ccs-syntax} and \ref{def:ccs-semantics}, based on~{\cite{lts-book}).

\begin{definition}[CCS syntax]\label{def:ccs-syntax}
    The syntax of CCS is:

    \smallskip\centerline{$\begin{array}{c}%
        \mu \mathrel{\;\Coloneqq\;} \tau \ |\ a \ | \ \overline{a}
        \qquad
        P \mathrel{\;\Coloneqq\;} \mathbf{0} \ | \ \mu.Q \ | \ P + P'
        \qquad
        Q \mathrel{\;\Coloneqq\;} P \,\ | \,\ Q \mathbin{|} Q' \,\ | \,\ (\nu a)Q \,\ | \,\ X
    \end{array}$}\smallskip
\end{definition}

By \autoref{def:ccs-syntax}, an \emph{action} $\mu$ can be the silent action $\tau$, a visible action $a$, or its \emph{co-action} $\overline{a}$. A \emph{sequential CCS process} $P$ can do nothing ($\mathbf{0}$), perform an \emph{action prefix} $\mu$ followed by $Q$ ($\mu.Q$), or perform a \emph{choice} ($P + P'$). A \emph{process $Q$} can be a sequential process $P$, a \emph{parallel composition} of two processes ($Q \mathbin{|} Q'$), a \emph{restriction} of action $a$ to scope $Q$ ($(\nu a)Q$), or a \emph{process name} $X$.

The LTS semantics of CCS is formalised in \autoref{def:ccs-semantics} below, where it is assumed that there is a partial map of \emph{defining equations} $\mathcal{D}$ from process names to processes, i.e.,~$\mathcal{D}(X) = Q$ means that $\mathcal{D}$ defines the name $X$ as process $Q$.

\begin{definition}[LTS semantics of CCS]
\label{def:ccs-semantics}
    The LTS of a CCS process $Q$ with defining equations $\mathcal{D}$, written $LTS(Q, \mathcal{D})$, has the least transition relation $\xrightarrow{\cdot}$ induced by the rules below:
    \setlength{\abovedisplayskip}{0pt}
    \setlength{\belowdisplayskip}{0pt}
    \begin{gather*}
        \begin{alignedat}{8}
            \begin{split}
                \text{\emph{\textsc{Pref}}} & \frac{}{\mu.Q \xrightarrow{\mu} Q}
                \\
                \text{\emph{\textsc{Cons}}} & \frac{Q \xrightarrow{\mu} Q'}{X \xrightarrow{\mu} Q'} \ \mathcal{D}(X) = Q
            \end{split}
            & \quad &
            \begin{split}
                \text{\emph{\textsc{Sum1}}} & \frac{Q_1 \xrightarrow{\mu} Q_1'}{Q_1 + Q_2 \xrightarrow{\mu} Q_1'}
                \\
                \text{\emph{\textsc{Sum2}}} & \frac{Q_2 \xrightarrow{\mu} Q_2'}{Q_1 + Q_2 \xrightarrow{\mu} Q_2'}
            \end{split}
            & \quad &
            \begin{split}
                \text{\emph{\textsc{Par1}}} & \frac{Q_1 \xrightarrow{\mu} Q_1'}{Q_1 \mathbin{|} Q_2 \xrightarrow{\mu} Q_1' \mathbin{|} Q_2}
                \\
                \text{\emph{\textsc{Par2}}} & \frac{Q_2 \xrightarrow{\mu} Q_2'}{Q_1 \mathbin{|} Q_2 \xrightarrow{\mu} Q_1 \mathbin{|} Q_2'}
            \end{split}
        \end{alignedat}
        \\
        \begin{alignedat}{5}
            \text{\emph{\textsc{Com}}} & \frac{Q_1 \xrightarrow{a} Q_1' \quad Q_2 \xrightarrow{\overline{a}} Q_2'}{Q_1 \mathbin{|} Q_2 \xrightarrow{\tau} Q_1' \mathbin{|} Q_2'}
            & \qquad &
            \text{\emph{\textsc{Res}}} & \frac{Q \xrightarrow{\mu} Q'}{(\nu a)Q \xrightarrow{\mu} (\nu a)Q'} \ \mu \neq a, \overline{a}
        \end{alignedat}
    \end{gather*}
\end{definition}

By rule \textsc{Pref} in \autoref{def:ccs-semantics},
actions ($a$), co-actions ($\overline{a}$), and internal actions ($\tau$) can be executed by consuming the prefix of a sequential process: for example, we have $a.\mathbf{0} \xrightarrow{a} \mathbf{0}$. The rules \textsc{Sum1} and \textsc{Sum2} allow for executing either the left or right branch of a choice: for example, we have $Q \xleftarrow{a} a.Q + b.Q' \xrightarrow{b} Q'$. By rule \textsc{Res}, actions and co-actions cannot be executed when restricted; for example, we have $(\nu b)(a.\mathbf{0}) \xrightarrow{a} (\nu b)\mathbf{0}$, whereas $b$ cannot be executed in $(\nu b)(b.\mathbf{0})$. By rule \textsc{Com}, an action can \emph{synchronize} with its co-action, producing an internal $\tau$-action; for example, $b$ can synchronize with $\overline{b}$, so we have $b.\mathbf{0} \mathbin{|} \overline{b}.\mathbf{0} \xrightarrow{\tau} \mathbf{0} \mathbin{|} \mathbf{0}$; this also works under restriction (by rule \textsc{Res}), so we have $(\nu b)(b.\mathbf{0} \mathbin{|} \overline{b}.\mathbf{0}) \xrightarrow{\tau} (\nu b)(\mathbf{0} \mathbin{|} \mathbf{0})$.

\paragraph{Petri nets.}
We adopt standard definitions of labelled Petri nets (that we simply call \emph{Petri nets}), marking, and firing rules (\autoref{def:petri-net}, \ref{def:marking}, and \ref{def:firing-rule}, based on \cite{alpha-miner}).
We also highlight two classes of Petri nets commonly used in process mining literature: \emph{workflow nets} (\autoref{def:workflow-net}) and \emph{free-choice nets} (\autoref{def:free-choice-net}).

\begin{definition}[Labelled Petri net]
\label{def:petri-net}
    A \emph{labelled Petri net} is a tuple $(P, T, F, A, \sigma)$ where $P$ is a finite set of \emph{places}, $T$ is a finite set of \emph{transitions} such that $P \cap T = \emptyset$, and $F \subseteq (P \times T) \cup (T \times P)$ is a set of directed \emph{edges} from places to transitions or \emph{vice versa}; moreover, $A$ is a set of \emph{actions} and $\sigma: T \to A$ assigns an action to each transition.
\end{definition}

Notably, \autoref{def:petri-net} only allows for at most one (unweighted) edge between each pair of places and transitions, and does \emph{not} allow co-actions in $A$: we keep co-actions exclusive to CCS (\autoref{def:ccs-syntax}) to avoid renaming in our encodings.

\begin{definition}[Marking]
\label{def:marking}
    A \emph{marking} of a Petri net $(P, T, F, A, \sigma)$ is a mapping $M: P \to \mathbb{N}$ from each place $p \in P$ to the number of \emph{tokens} in $p$ (may be 0).
\end{definition}

\begin{definition}[Firing rule]
\label{def:firing-rule}
    Given a Petri net $(P, T, F, A, \sigma)$ and a marking $M: P \to \mathbb{N}$, a transition $t \in T$ is \emph{enabled} if all places with an edge to $t$ have tokens in $M$. Transition $t$ can \emph{fire} if enabled, and this firing consumes one token from all places with an edge to $t$, emits a label $\sigma(t)$, and produces one token for all places with an edge from $t$. This results in an updated marking $M'$.
\end{definition}

A Petri net $N = (P, T, F, A, \sigma)$ and an initial marking $M_0$ yields $LTS(N, M_0) = (Q, A, \xrightarrow{\cdot})$ where the states ($Q$) and the transition relation ($\xrightarrow{\cdot}$) are derived using the firing rule (\autoref{def:firing-rule}) until all enabled transitions $t$ in all reachable markings are added to $\xrightarrow{\cdot}$, with $\sigma(t)$ as transition label.

Process mining algorithms like the $\alpha$-miner typically produce \emph{workflow nets} (\autoref{def:workflow-net}) that are used to describe end-to-end processes with a clear start and completion. In practical applications, such workflow nets are often \emph{free-choice}\footnote{
    The $\alpha$-miner actually returns a further subclass of free-choice workflow nets called \emph{structured workflow nets}~\cite{alpha-miner}.
} (\autoref{def:free-choice-net} and \ref{def:free-choice-workflow-net}) where all choices are made by a single place --- meaning that transitions can at most fight for one token in order to fire. \autoref{fig:preliminaries-none}--\ref{fig:preliminaries-free-choice-workflow} show examples of Petri nets with small differences resulting in different classes.

\begin{definition}[Workflow net~{\cite[Definition 2.8]{alpha-miner}}]\label{def:workflow-net}
    A Petri net $(P, T, F, A, \sigma)$ is a \emph{workflow net} iff it satisfies the following three properties:
    \begin{itemize}
        \lessspace
        \item \textbf{Object creation}: $P$ has an \emph{input place} $i$ with no ingoing edges.
        \item \textbf{Object completion}: $P$ has an \emph{output place} $o$ with no outgoing edges.
        \item \textbf{Connectedness}: For every $v \in P \cup T$, there exists a directed path of edges from $i$ to $o$ that goes through $v$.
    \end{itemize}
\end{definition}

\begin{definition}[Free-choice net~{\cite{free-choice-nets}}]\label{def:free-choice-net}
    A Petri net $(P, T, F, A, \sigma)$ is a \emph{free-choice net} iff it satisfies the following two properties:
    \begin{itemize}
        \lessspace
        \item \textbf{Unique choice}: All places $p \in P$ with more than one outgoing edge only have edges to transitions with exactly one ingoing edge (edge from $p$).
        \item \textbf{Unique synchronisation}: All transitions $t \in T$ with more than one ingoing edge only have edges from places with exactly one outgoing edge (edge to $t$).
        \lessspace
    \end{itemize}
\end{definition}

\begin{definition}[Free-choice workflow net]\label{def:free-choice-workflow-net}
    If a Petri net is both a workflow net and free-choice, we call it a \emph{free-choice workflow net}.
\end{definition}

\begin{figure}[tb!]
    \centering
    \begin{minipage}[b][][b]{0.49\textwidth}
        \centering
        \includegraphics[scale=0.30]{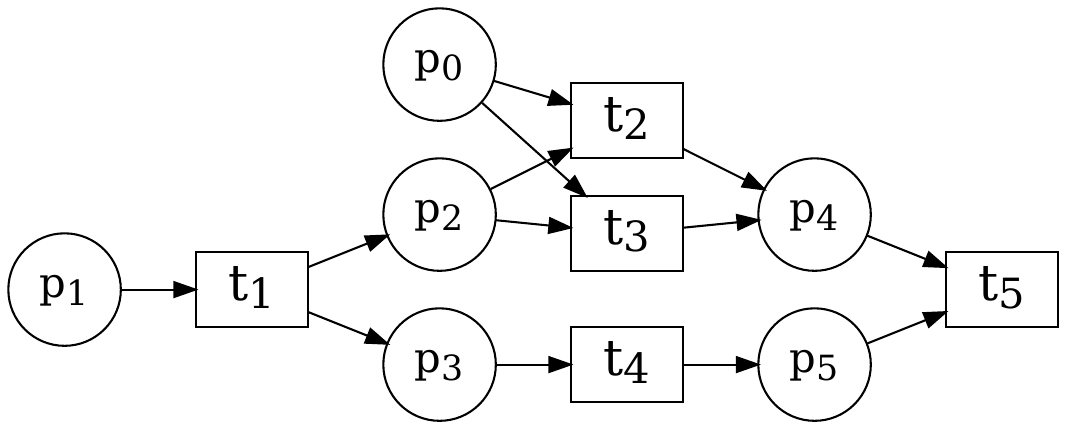}
        \caption{Neither free-choice or workflow net.}
        \label{fig:preliminaries-none}
    \end{minipage}\hfill%
    \begin{minipage}[b][][b]{0.49\textwidth}
        \centering
        \makebox[\textwidth][c]{\includegraphics[scale=0.30]{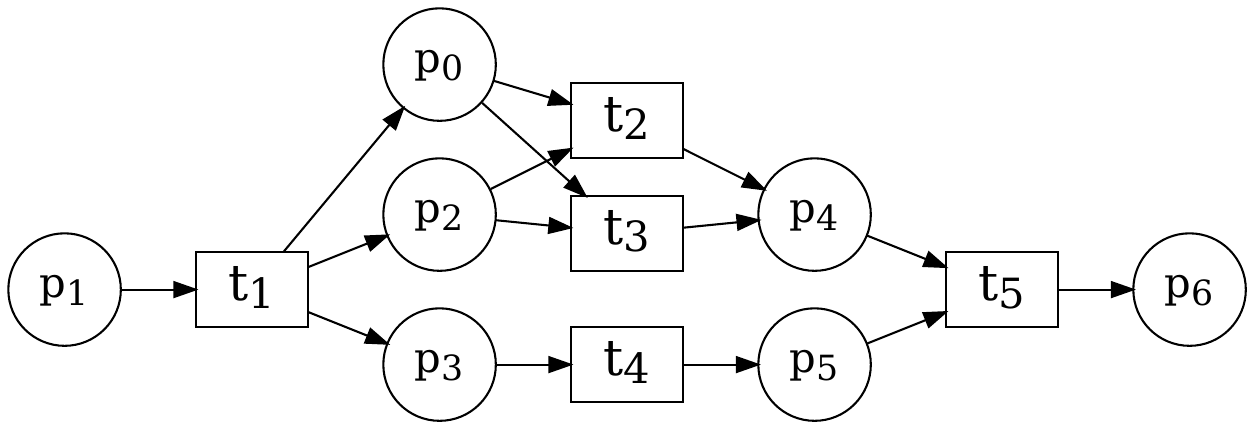}}
        \caption{Workflow but not free-choice net.}
        \label{fig:preliminaries-workflow}
    \end{minipage}
\end{figure}
\begin{figure}[tb!]
    \centering
    \begin{minipage}[b][][b]{0.49\textwidth}
        \centering
        \includegraphics[scale=0.30]{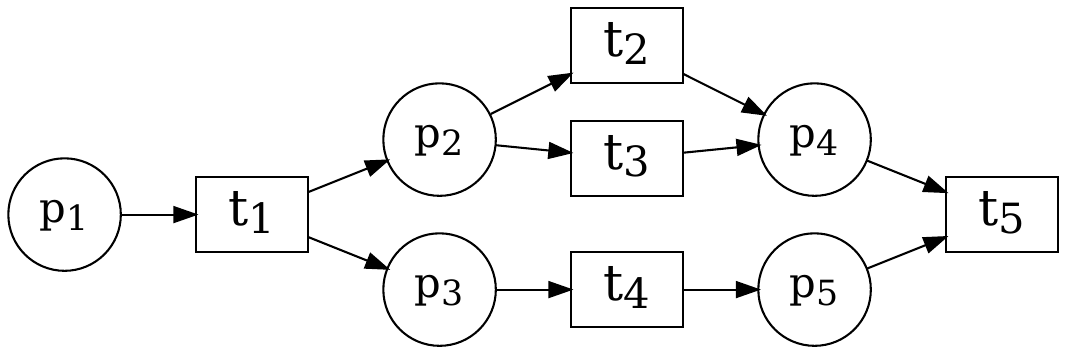}
        \caption{Free-choice but not workflow net.}
        \label{fig:preliminaries-free-choice}
    \end{minipage}\hfill%
    \begin{minipage}[b][][b]{0.49\textwidth}
        \centering
        \makebox[\textwidth][c]{\includegraphics[scale=0.30]{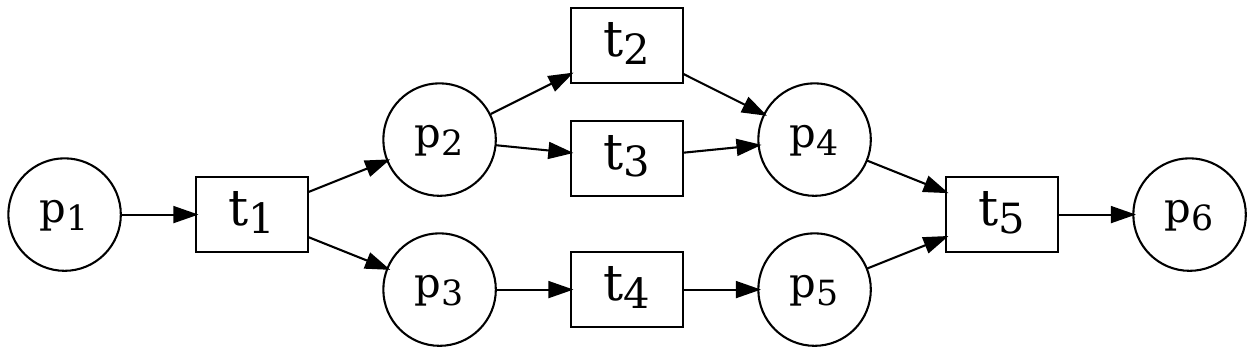}}
        \caption{Free-choice net and workflow net.}
        \label{fig:preliminaries-free-choice-workflow}
    \end{minipage}
\end{figure}

The \emph{connectedness} property in \autoref{def:workflow-net} implies that all transitions in free-choice workflow nets have at least one ingoing edge and at least one outgoing edge. The same applies to all places except the special places $i$ and $o$ that respectively has no ingoing edges ($i$) and no outgoing edges ($o$). The \emph{unique choice and synchronisation} properties in \autoref{def:free-choice-net} ensure that every choice is separated from all synchronisations and \emph{vice versa}. This does \emph{not} rule out cycles but \emph{unique choice} restricts all outgoing edges \emph{leaving} a cycle to lead to transitions with exactly one ingoing edge (because there is always one edge \emph{continuing} the cycle). For instance, adding a new transition $t_7$ in \autoref{fig:preliminaries-free-choice-workflow} to form the cycle $(p_3, t_4, p_5, t_7, p_3)$ would violate \emph{unique choice} (and \emph{unique synchronisation}) because the edge $(p_5, t_5)$ leaves the cycle but $t_5$ has two ingoing edges. Adding the cycle $(p_3, t_7, p_3)$ is allowed because $t_4$ only has one ingoing edge.

\section{Encoding Petri Nets into CCS, Step-by-Step}\label{sec:encoding}

This section introduces our main contribution: an encoding into CCS of a superclass of free-choice nets, that we call \emph{group-choice nets} (\autoref{def:group-choice-net}); we prove that our encoding is correct, i.e., a Petri net and its encoding are weakly bisimilar and without added divergent states (\autoref{thm:group}). To illustrate the encodings and result, we proceed through a series of steps: a series of encoding algorithms into CCS for progressively larger classes of Petri nets. (See \autoref{fig:overview} for an outline.)

We begin (in~\autoref{sec:ccs-net}) with the class of \emph{CCS nets} (\autoref{def:ccs-net}), and \autoref{alg:ccs-net} that encodes such nets into strongly bisimilar CCS processes (\autoref{thm:ccs-net-strong-bisimulation}).\footnote{%
The results in~\autoref{sec:ccs-net} can be also derived from \cite{multi-ccs}, but here we provide a direct encoding algorithm, statements, and proofs for plain CCS, without using Multi-CCS.}%
Then (in~\autoref{sec:free-wf}) we develop a novel transformation from \emph{free-choice workflow nets} (\autoref{def:free-choice-workflow-net}) to weakly bismilar CCS nets using \autoref{alg:free-wf-full} (\autoref{thm:free-wf-full-weak-bisimulation}). The composition of \autoref{alg:free-wf-full} and \autoref{alg:ccs-net} then encodes free-choice workflow nets into weakly bismilar CCS processes (\autoref{thm:free-wf-full}). In \autoref{sec:free}, we generalise CCS nets into \emph{2-$\tau$-synchronisation nets} (\autoref{def:2-tau-synchronisation-net}, allowing for transitions with no ingoing edges) and we present \autoref{alg:free} to encode such nets into strongly bisimilar CCS processes (\autoref{thm:free}). Finally (cf.~\autoref{sec:group}), we generalize free-choice nets to a new class called \emph{group-choice nets} (\autoref{def:group-choice-net}) and present \autoref{alg:group-full} that, composed with \autoref{alg:free}, encodes group-choice nets into weakly bisimilar CCS processes (\autoref{thm:group}).

\subsection{Encoding CCS Nets into CCS Processes}\label{sec:ccs-net}

A challenge in encoding Petri nets into CCS processes is that a transition in a Petri net can consume tokens from any number of places in a single step --- whereas the semantics of CCS (\autoref{def:ccs-semantics}) only allow for executing a single action or a synchronisation between two processes in each step. In other words, Petri nets are able to perform $n$-ary synchronisation, while CCS can only perform $2$-ary synchronisation.
Therefore, as a stepping stone towards our main result, we adopt \emph{CCS nets} from~\cite{multi-ccs}, whose synchronisation capabilities match CCS.

\begin{definition}[CCS net~\cite{multi-ccs}]\label{def:ccs-net}
 A Petri net $(P, T, F, A, \sigma)$ is a \emph{CCS net} iff each transition $t \in T$ has one or two ingoing edges --- and in the latter case, $\sigma(t) = \tau$.
\end{definition}

The key insight behind \autoref{def:ccs-net} is that transitions with two ingoing edges must be labeled with $\tau$ to have a 1-to-1 correspondence to synchronisation in CCS (that also emits a $\tau$). \autoref{fig:ccs-net-example-pn} shows a CCS net that later will be encoded.

\begin{figure}[b]
    \setlength{\abovedisplayskip}{0pt}
    \setlength{\belowdisplayskip}{0pt}
    \centering
    \begin{minipage}[b][][b]{0.32\textwidth}
        \centering
        \vfill
        \includegraphics[scale=0.30]{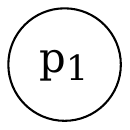}
        \vspace{4pt}
        \vfill
        \begin{align*}
            X_{p_1} \mathrel{:=} \mathbf{0}
        \end{align*}
        \caption{Place with $0$ outgoing edges and its encoding.}
        \label{fig:ccs-net-place-0}
    \end{minipage}\hfill%
    \begin{minipage}[b][][b]{0.32\textwidth}
        \centering
        \vfill
        \includegraphics[scale=0.30]{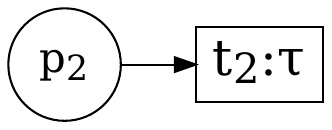}
        \vspace{4pt}
        \vfill
        \begin{align*}
            X_{p_2} \mathrel{:=} Y_{t_2}
        \end{align*}
        \caption{Place with $1$ outgoing edge and its encoding.}
        \label{fig:ccs-net-place-1}
    \end{minipage}\hfill%
    \begin{minipage}[b][][b]{0.32\textwidth}
        \centering
        \includegraphics[scale=0.30]{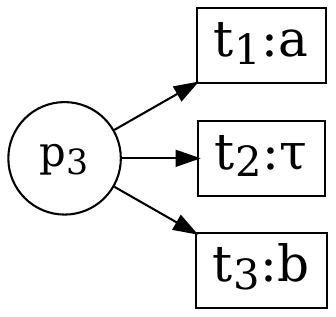}
        \begin{align*}
            X_{p_3} \mathrel{:=} Y_{t_1} + Y_{t_2} + Y_{t_3}
        \end{align*}
        \caption{Place with $3$ outgoing edges and its encoding.}
        \label{fig:ccs-net-place-3}
    \end{minipage}
\end{figure}

\begin{figure}[t]
    \setlength{\abovedisplayskip}{0pt}
    \setlength{\belowdisplayskip}{0pt}
    \centering
    \begin{minipage}[b][][b]{0.32\textwidth}
        \centering
        \vfill
        \includegraphics[scale=0.30]{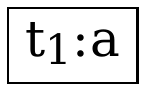}
        \vspace{5pt}
        \vfill
        \begin{align*}
            Y_{t_1} \mathrel{:=} a.\mathbf{0}
        \end{align*}
        \caption{Encoding of transition with $0$ outgoing edges.}
        \label{fig:ccs-net-transition-0}
    \end{minipage}\hfill%
    \begin{minipage}[b][][b]{0.32\textwidth}
        \centering
        \vfill
        \includegraphics[scale=0.30]{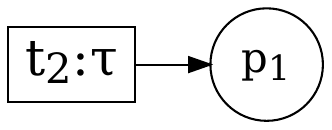}
        \vspace{2pt}
        \vfill
        \begin{align*}
            Y_{t_2} \mathrel{:=} \tau.X_{p_1}
        \end{align*}
        \caption{Encoding of transition with $1$ outgoing edge.}
        \label{fig:ccs-net-transition-1}
    \end{minipage}\hfill%
    \begin{minipage}[b][][b]{0.32\textwidth}
        \centering
        \includegraphics[scale=0.30]{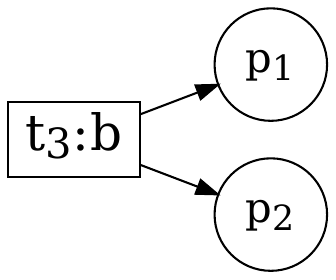}
        \begin{align*}
            Y_{t_3} \mathrel{:=} b.(X_{p_1} \mathbin{|} X_{p_2})
        \end{align*}
        \caption{\hspace{-1pt}Encoding of transition with $2$ outgoing edges.}
        \label{fig:ccs-net-transition-2}
    \end{minipage}
\end{figure}

\begin{algorithm}[!b]
    \DontPrintSemicolon
    \SetKwInOut{Input}{Input}
    \SetKwInOut{Output}{Output}
    \SetKw{Replace}{substitute}
    \SetKw{In}{in}
    \SetKw{With}{with}
    \SetKw{Where}{where}

    \Input{CCS net $(P, T, F, A, \sigma)$ and marking $M_0: P \to \mathbb{N}$}
    \Output{CCS process $Q$ and partial mapping of defining equations $\mathcal{D}$}

    $\mathcal{D} \gets \text{Empty mapping of defining equations}$\;
    \For{$p \in P$}{
        $\mathcal{D}(X_p) \gets \left(Y_{t_1} + Y_{t_2} + \dots + Y_{t_k}\right)$ \Where $\{t_1, t_2, \dots, t_k\} = \{t \,|\, (p, t) \in F\}$\;
    }
    \For{$t \in \left\{t \;\middle|\; t \in T \text{\emph{ and }} t \text{\emph{ has 1 ingoing edge}}\right\}$}{
        \Replace $Y_t$ \In $\mathcal{D}(X_{p^*})$ \With $\sigma(t).(X_{p_1} \mathbin{|} X_{p_2} \mathbin{|} \dots \mathbin{|} X_{p_k})$ \Where $(p^*, t) \in F \text{ and } \{p_1, p_2, \dots, p_k\} = \{p \,|\, (t, p) \in F\}$\;
    }
    $A' \gets \emptyset$\;
    \For{$t \in \left\{t \;\middle|\; t \in T \text{\emph{ and }} t \text{\emph{ has 2 ingoing edges}}\right\}$}{
        $A' \gets A' \cup \{s_t\}$ \Where $s_t \text{ is a fresh action}$\;
        \Replace $Y_t$ \In $\mathcal{D}(X_{p^*})$ \With $s_t.(X_{p_1} \mathbin{|} X_{p_2} \mathbin{|} \dots \mathbin{|} X_{p_k})$\;
        \Replace $Y_t$ \In $\mathcal{D}(X_{p^{**}})$ \With $\overline{s_t}.\mathbf{0}$ \Where $(p^*, t), (p^{**}, t) \in F \text{ and } p^* \neq p^{**} \text{ and } \{p_1, p_2, \dots, p_k\} = \left\{p \,\middle|\, (t, p) \in F\right\}$\;
    }
    $\{s_1, s_2, \dots, s_n\} \gets A'$\;
    $Q \gets (\nu s_1) (\nu s_2) \dots (\nu s_n) \left({X_{p_1}^{M_0(p_1)}} \mathbin{|} {X_{p_2}^{M_0(p_2)}} \mathbin{|} \dots \mathbin{|} {X_{p_{|P|}}^{M_0(p_{|P|})}}\right)$\;
    \Return{$\left(Q,\; \mathcal{D}\right)$}
    \caption{Encoding from CCS net to CCS process}
    \label{alg:ccs-net}
\end{algorithm}

\begin{figure}[t]
    \setlength{\abovedisplayskip}{0pt}
    \setlength{\belowdisplayskip}{0pt}
    \centering
    \begin{minipage}[b][][b]{0.49\textwidth}
        \centering
        \includegraphics[scale=0.30]{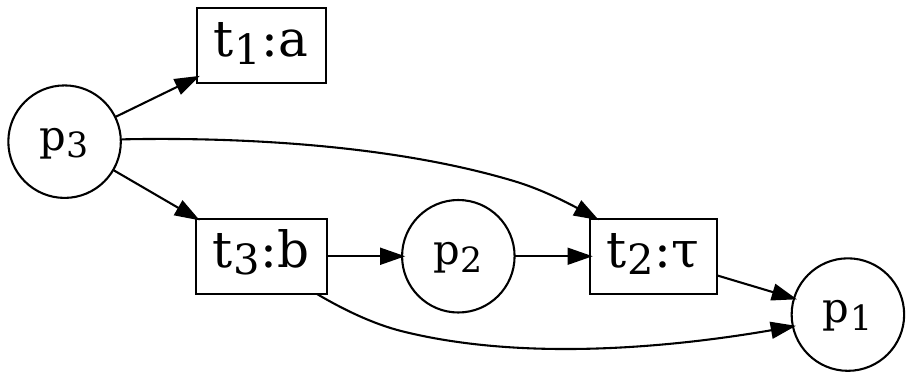}
        \begin{align*}
            M_0(p_1) = 1,\; M_0(p_2) = 0,\; M_0(p_3) = 2
        \end{align*}
        \caption{CCS net $N$ and initial marking $M_0$ with three tokens.}
        \label{fig:ccs-net-example-pn}
    \end{minipage}\hfill%
    \begin{minipage}[b][][b]{0.49\textwidth}
        \centering
        \begin{align*}
            \begin{array}{rcl}
                Q     & \mathrel{\coloneqq} & (\nu s_{t_2})(X_{p_1} \mathbin{|} X_{p_3} \mathbin{|} X_{p_3}) \\
                \mathcal{D}(X_{p_1}) & \mathrel{=} & \mathbf{0} \\
                \mathcal{D}(X_{p_2}) & \mathrel{=} & \overline{s_{t_2}}.\mathbf{0} \\
                \mathcal{D}(X_{p_3}) & \mathrel{=} & s_{t_2}.X_{p_1} + a.\mathbf{0} + b.(X_{p_1} \mathbin{|} X_{p_2})
            \end{array}
        \end{align*}
        \caption{Encoding of $N$ and $M_0$ in \autoref{fig:ccs-net-example-pn} into $Q$ and $\mathcal{D}$ produced by \autoref{alg:ccs-net}.}
        \label{fig:ccs-net-example-ccs}
    \end{minipage}
\end{figure}

\autoref{alg:ccs-net} encodes a CCS net $(P, T, F, A, \sigma)$ and marking $M_0$ into a CCS process $Q$ and its defining equations $\mathcal{D}$, where $\mathcal{D}$ defines a process named $X_p$ for each place $p \in P$. The idea is that each token at place $p$ is encoded as a parallel replica of the process $X_p$. We illustrate the algorithm in the next paragraphs.

On line 2--4, each place $p \in P$ is encoded as a \emph{choice process} named $X_p$ in $\mathcal{D}$: The choice is among placeholder processes named $Y_t$, for each transition $t$ with an edge from $p$. (Notice that the placeholders $Y_t$ are not in the domain of $\mathcal{D}$, but are substituted with sequential CCS processes in the next steps of the algorithm.) The choice process $X_p$ models a token at $p$ that \emph{chooses} which transition it is used for. \autoref{fig:ccs-net-place-0}--\ref{fig:ccs-net-place-3} show examples.

On line 5--7, each transition $t \in T$ with one ingoing edge from a place $p^*$ is encoded as a process with an action prefix (obtained via the labelling function $\sigma(t)$) followed by the \emph{parallel composition} of all processes named $X_{p_i}$, where place $p_i$ has an ingoing edge from $t$. \autoref{fig:ccs-net-transition-0}--\ref{fig:ccs-net-transition-2} show examples. The resulting process $\sigma(t).(X_{p_1} \mathbin{|} \dots \mathbin{|} X_{p_k})$ (line 6) models the execution of the action $\sigma(t)$ followed by the production of tokens for places $p_1, \dots, p_k$, and such a process is used to substitute the placeholder $Y_t$ in $\mathcal{D}(X_{p^*})$.

On line 8--13, each transition $t \in T$ with two ingoing edges from $p^*$ and $p^{**}$ (where $t$ has label $\sigma(t) = \tau$, by \autoref{def:ccs-net}) is encoded as a synchronisation between a fresh action $s_t$ (which prefixes $(X_{p_1} \mathbin{|} \dots \mathbin{|} X_{p_k})$ on line 11 to model production of tokens), and its co-action $\overline{s_t}$ (line 12). The two resulting processes are respectively used to substitute the placeholder $Y_t$ in $\mathcal{D}(X_{p^*})$ and $\mathcal{D}(X_{p^{**}})$.

Finally, on line 15, the initial marking $M_0$ is encoded into the result CCS process $Q$, which is the \emph{parallel composition} of one place process $X_{p_i}$ per token at place $p_i$, under a restriction of each fresh action $s_i$ produced on line 10. Note that we write $Q^n$ for the parallel composition of $n$ replicas of $Q$, so $Q^3 = Q \mathbin{|} Q \mathbin{|} Q$, and $Q^1 = Q$, and $Q^0 = \mathbf{0}$. The algorithm also returns the partial mapping $\mathcal{D}$ with the definition of each process name $X_p$ (for each $p \in P$) occurring in $Q$.

\autoref{fig:ccs-net-example-ccs} shows the encoding of the CCS net in \autoref{fig:ccs-net-example-pn}.
Observe that in $X_{p_3}$, the transitions $t_1$ and $t_3$ (which have one ingoing edge) yield respectively the sub-processes $a.\mathbf{0}$ and $b.(X_{p_1} \mathbin{|} X_{p_2})$ (produced by line 5--7 in \autoref{alg:ccs-net}); instead, $t_2$ (which has two ingoing edges and label $\tau$) yields both sub-processes $s_{t_2}.X_{p_1}$ in $X_{p_3}$ and $\overline{s_{t_2}}.\mathbf{0}$ in $X_{p_2}$ (produced by line 9--13 in \autoref{alg:ccs-net}).

\begin{theorem}[Correctness of \autoref{alg:ccs-net}]\label{thm:ccs-net-strong-bisimulation}
    Given a CCS net $N = (P, T, F, A, \sigma)$ and an initial marking $M_0$, let the result of applying \autoref{alg:ccs-net} to $N$ and $M_0$ be the CCS process $Q$ and defining equations $\mathcal{D}$. Then, $LTS(N, M_0) \bisimilar LTS(Q, \mathcal{D})$. The translation time and the size of $(Q, \mathcal{D})$ are bound by $\O(|N| + \sum_{p \in P} M_0(p))$.
\end{theorem}
\begin{proof}
    \emph{(Sketch, detailed proof in \appendixref{sec:appendix-ccs-net-strong-bisimulation}.)} We define the following relation between markings and processes, and we prove it is a bisimulation:
    \begin{align*}
        \mathcal{R} \mathrel{\;\coloneqq\;} \left\{\left(M,\ (\nu s_1)\dots(\nu s_n)\left({X_{p_1}^{M(p_1)}} \mathbin{|} \dots \mathbin{|} {X_{p_{|P|}}^{M(p_{|P|})}}\right)\right) \;\middle|\; M \in LTS(N, M_0)\right\}
    \end{align*}
    where $M$ is a reachable marking from the initial marking $M_0$ in $N$ (i.e.~$M$ is a state in $LTS(N, M_0)$), and $s_1, \dots, s_n$ are the restricted names in $Q$. The intuition is as that for every place $p$ containing $M(p)$ tokens, the CCS process contains $M(p)$ replicas of the process named $X_p$ (written $X_p^{M(p)}$).
    \qed
\end{proof}

The strong bisimulation result in \autoref{thm:ccs-net-strong-bisimulation} ensures that there are no observable differences between the original Petri net and its encoding, so \autoref{alg:ccs-net} does not introduce new deadlocks nor divergent paths. The encoding is also linear in the size of the CCS net $N$ and the size of the initial marking $M_0$.

\subsection{Encoding Free-Choice Workflow Nets into CCS Processes}\label{sec:free-wf}

We now build upon the result in \autoref{sec:ccs-net} to prove that \emph{free-choice workflow nets} (\autoref{def:free-choice-workflow-net}) are encodable into weakly bisimilar CCS processes. Specifically, we present a stepwise transformation procedure (\autoref{alg:free-wf-stepwise} and \ref{alg:free-wf-full}) from a free-choice workflow net into weakly bisimilar CCS net (\autoref{thm:free-wf-full-weak-bisimulation}), and then apply \autoref{alg:ccs-net} to get a weakly bisimilar CCS process (\autoref{thm:free-wf-full}).

It should be noted (as illustrated in \autoref{fig:overview}) that free-choice (workflow) nets does \emph{not} contain the class of CCS nets: \autoref{fig:ccs-net-example-pn} is not a free-choice net since $p_3$ has multiple outgoing edges of which one leads to a transition with multiple ingoing edges. However, as will be shown in this section, a free-choice workflow net can be transformed into a CCS net by (non-deterministically) making a binary synchronisation order for each $n$-ary synchronisation.

\begin{algorithm}[b]
    \DontPrintSemicolon
    \SetKwInOut{Input}{Input}
    \SetKwInOut{Output}{Output}
    \SetKw{Select}{select}
    \SetKw{Where}{where}

    \Input{Petri net $(P, T, F, A, \sigma)$, marking $M_0: P \to \mathbb{N}$ and chosen transition $t^* \in T$ with at least two ingoing edges}
    \Output{Petri net $(P^\prime, T^\prime, F^\prime, A^\prime, \sigma^\prime)$ and marking $M_0^\prime:P^\prime \to \mathbb{N}$}

    \Select $(p^*, t^*), (p^{**}, t^*) \in F$ \Where $p^* \neq p^{**}$\;
    $P' \gets P \cup \{p^+\}$ \Where $p^+ \notin P \cup T$\;
    $T' \gets T \cup \{t^+\}$ \Where $t^+ \notin P' \cup T$\;
    $F' \gets (F \setminus \{(p^*, t^*), (p^{**}, t^*)\}) \cup \{(p^*, t^+), (p^{**}, t^+), (t^+, p^+), (p^+, t^*)\}$\;
    $N' \gets (P', T', F', A \cup \{\tau\}, \sigma[t^+ \mapsto \tau])$\;
    \Return{$(N', M_0[p^+ \mapsto 0])$}\;

    \caption{Petri net transition preset reduction}
    \label{alg:free-wf-stepwise}
\end{algorithm}

\begin{figure}[!b]
    \centering
    \begin{minipage}[b][][b]{0.44\textwidth}
        \centering
        \includegraphics[scale=0.30]{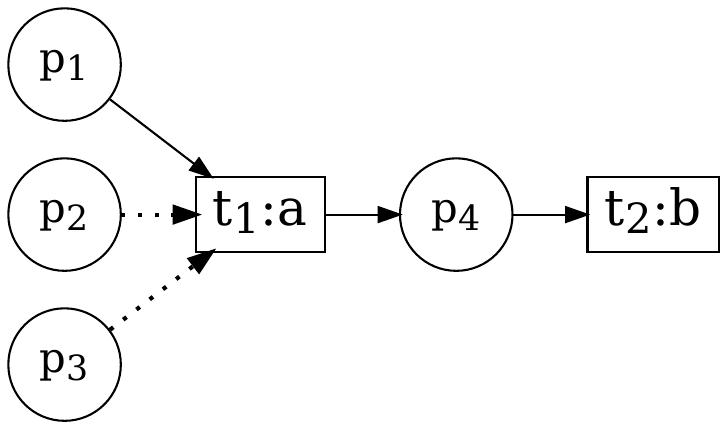}
        \caption{Free-choice net where $t_1$ and the dotted edges $(p_2, t_1)$ and $(p_3, t_1)$ are chosen for transformation.}
        \label{fig:free-wf-example-before}
    \end{minipage}\hfill%
    \begin{minipage}[b][][b]{0.54\textwidth}
        \centering
        \includegraphics[scale=0.30]{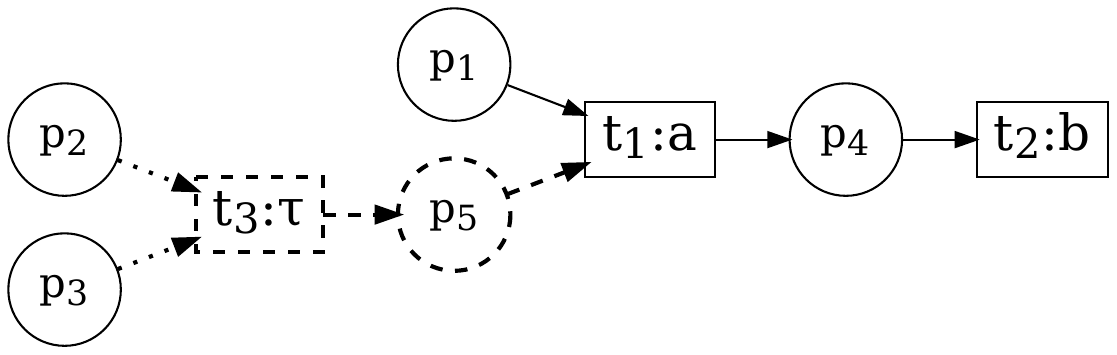}
        \caption{Free-choice net obtained from \autoref{fig:free-wf-example-before} where the dotted edges have been moved and the dashed elements $t_3$ and $p_5$ have been added.}
        \label{fig:free-wf-example-after}
    \end{minipage}
\end{figure}

An application of \autoref{alg:free-wf-stepwise} transforms an input Petri net by reducing the number of ingoing edges of the selected transition $t^*$ (which must have two or more ingoing edges): it selects two distinct ingoing edges (line 1), creates a new $\tau$-transition $t^+$ with an edge to a new place $p^+$, and connects the new place $p^+$ to the selected transition $t^*$ (line 2--5). The cost of the transformation is that a new $\tau$-transition $t^+$ with two ingoing edges is created. \autoref{fig:free-wf-example-before} and \ref{fig:free-wf-example-after} show an example application on a free-choice net, where $t_1$ is selected for transformation with the edges $(p_2, t_1)$ and $(p_3, t_1)$. \autoref{alg:free-wf-stepwise} can be iterated until all the original transitions only have one ingoing edge left; transitions with no ingoing edges are left untouched. Hence, transformed free-choice workflow nets will have no transitions with zero ingoing edges, hence \autoref{alg:ccs-net} could be applied.

\autoref{fig:free-wf-example-before} and \ref{fig:free-wf-example-after} show that \autoref{alg:free-wf-stepwise} does \emph{not} produce a Petri net that is strongly bisimilar to its input: if we only have one token in both $p_2$ and $p_3$, then \autoref{fig:free-wf-example-after} can fire the $\tau$-transition while \autoref{fig:free-wf-example-before} is in a deadlock. However, \autoref{thm:free-wf-stepwise-weak-bisimulation} below proves that, when applied to free-choice nets (and thereby also free-choice workflow nets), \autoref{alg:free-wf-stepwise} produces a weakly bisimilar Petri net. Moreover, by \autoref{lem:free-wf-stepwise-invariant}, the returned net remains a free-choice net.

\begin{theorem}[Correctness of \autoref{alg:free-wf-stepwise}]\label{thm:free-wf-stepwise-weak-bisimulation}
    Given a free-choice net $N = (P, T, F, A, \sigma)$, a marking $M_0: P \to \mathbb{N}$, and a transition $t^* \in T$ (with at least two ingoing edges), the result of applying \autoref{alg:free-wf-stepwise} on $N$, $M_0$, and $t^*$ is a Petri net $N'$ and marking $M_0'$ such that $LTS(N, M_0) \wbisimilar LTS(N', M_0')$ and $LTS(N', M_0')$ contains a divergent path iff $LTS(N, M_0)$ contains a divergent path. The transformation time and the increase in size are amortized $\O(1)$.
\end{theorem}
\begin{proof}
    \emph{(Sketch, detailed proof in  \appendixref{sec:appendix-free-wf-stepwise-weak-bisimulation}.)} We define the following relation between pair of markings, and we prove it is a bisimulation:
    \begin{align*}
        \mathcal{R} \mathrel{\;\coloneqq\;} \left\{\left(M,\ M[p^* \mapsto M(p^*) - i][p^{**} \mapsto M(p^{**}) - i][p^+ \mapsto i]\right) \middle|
        \begin{smallmatrix}
            M \in LTS(N, M_0) \\
            k = \min(M(p^*), M(p^{**})) \\
            0 \leq i \leq k
        \end{smallmatrix}
        \right\}
    \end{align*}
    where $M$ is a reachable marking from the initial marking $M_0$ in the free-choice net $N$ ($M$ is a state in $LTS(N, M_0)$). The intuition behind the relation is: when $i = 0$, it covers the direct extension of a marking in $N$ to one in $N'$ where the new place $p^+$ has no tokens; when $i > 0$, it covers the cases where the new transition $t^+$ has been fired $i$ times without firing $t^*$ afterwards.

    For instance, consider \autoref{fig:free-wf-example-before} and \ref{fig:free-wf-example-after} where $t^* = t_1$, $p^* = p_2$, $p^{**} = p_3$, $t^+ = t_3$ and $p^+ = p_5$. Clearly, the initial markings $(M_0, M_0')$ are in $\mathcal{R}$ since \autoref{alg:free-wf-stepwise} just extends $M_0$ to $M_0'$ where $p^+$ has zero tokens. Now consider any pair $(M, M') \in \mathcal{R}$. The transition $t_2$ is not part of the transformed part, so $t_2$ can be fired in $N$ iff it can be fired in $N'$. Firing $t_1$ in $N$ can be replied in $N'$ by firing $t_1$ when there is a token in $p_5$ ($i > 0$) and otherwise by firing $t_3$ followed by $t_1$ ($i = 0$). Firing $t_1$ in $N'$ can be replied in $N$ by also firing $t_1$. If $t_3$ is fired in $N'$ then $N$ should do nothing as this simply merges the tokens from $p_2$ and $p_3$ into one token in $p_5$ by a $\tau$-action (increases $i$ by one). All these cases result in a new pair in $\mathcal{R}$ which completes the proof of $LTS(N, M_0) \wbisimilar LTS(N', M_0')$.

    In general, weak bisimulation alone does not ensure the absence of new divergent paths~\cite{lts-book}. However, \autoref{alg:free-wf-stepwise} only extends existing paths to $t^*$ by one $\tau$-transition such that no divergent paths are changed, so $LTS(N', M_0')$ contains a divergent path iff $LTS(N, M_0)$ contains a divergent path.
    \qed
\end{proof}

\begin{lemma}[Invariant of \autoref{alg:free-wf-stepwise}]\label{lem:free-wf-stepwise-invariant}
    If \autoref{alg:free-wf-stepwise} is given a free-choice net $N$, it returns a free-choice net $N'$ that has no additional or changed transitions with no ingoing edges compared to $N$. \emph{(Proof available in \appendixref{sec:appendix-free-wf-stepwise-invariant}.)}
\end{lemma}

\begin{algorithm}[t]
    \DontPrintSemicolon
    \SetKwInOut{Input}{Input}
    \SetKwInOut{Output}{Output}
    \SetKw{Where}{where}

    \Input{Free-choice net $N = (P, T, F, A, \sigma)$ and marking $M_0: P \to \mathbb{N}$}
    \Output{Petri net $(P', T', F', A', \sigma')$ and marking $M_0': P' \to \mathbb{N}$}

    $(N', M_0') \gets (N, M_0)$\;
    \While{$\exists t^* \in T'$ \Where $|\{p \,|\, (p, t^*) \in F\}| > (\text{\leIf*{$\sigma(t^*) = \tau$}{$2$}{$1$}})$}{
        $(N', M_0') \gets \text{\autoref{alg:free-wf-stepwise}}(N', M_0', t^*)$\;
    }
    \Return{$(N', M_0')$}\;

    \caption{Iterative Petri net transition preset reduction}
    \label{alg:free-wf-full}
\end{algorithm}

To achieve our encoding of free-choice workflow into CCS nets we use \autoref{alg:free-wf-full}, which applies \autoref{alg:free-wf-stepwise} (by non-deterministically picking a transition $t^*$) until all transitions have at most one ingoing edge (or two for $\tau$-transitions). If \autoref{alg:free-wf-full} is applied to a free-choice workflow net, it returns a CCS net (\autoref{lem:free-wf-full-correctness}) that is weakly bisimilar and has no new divergent paths (\autoref{thm:free-wf-full-weak-bisimulation}).

\begin{lemma}[\autoref{alg:free-wf-full} output]\label{lem:free-wf-full-correctness}
    If applied to a free-choice workflow net $(P, T, F, A, \sigma)$, \autoref{alg:free-wf-full} returns a CCS net.
    \emph{(Proof: \appendixref{sec:appendix-free-wf-full-correctness}.)}
\end{lemma}

\begin{theorem}[Correctness of \autoref{alg:free-wf-full}]\label{thm:free-wf-full-weak-bisimulation}
    Given a free-choice net $N = (P, T, F, A, \sigma)$ and a marking $M_0: P \to \mathbb{N}$, the result of applying \autoref{alg:free-wf-full} on $N$ and $M_0$ is a Petri net $N'$ and marking $M_0'$ such that $LTS(N, M_0) \wbisimilar LTS(N', M_0')$ and $LTS(N, M_0)$ has a divergent path iff $LTS(N', M_0')$ has. Both the transformation time and the size of $(N', M')$ are $\O(|N|)$.
\end{theorem}
\begin{proof}
    \emph{(Sketch, detailed proof in \appendixref{sec:appendix-free-wf-full-weak-bisimulation}.)} Follows from induction on the number of applications of \autoref{alg:free-wf-stepwise} and transitivity of bisimulation~\cite{lts-book}.
    \qed
\end{proof}

The full encoding from a free-choice workflow net into a weakly bisimilar CCS process is obtained by composing \autoref{alg:free-wf-full} and \autoref{alg:ccs-net}, leading to \autoref{thm:free-wf-full}. \autoref{alg:free-wf-stepwise} can at most be applied $|F|$ times. It runs in constant time and increases the size by a constant amount such that \autoref{alg:free-wf-full} produces a linear-sized CCS net. Hence, the resulting CCS process is also linear in size of the original Petri net and marking.

\begin{theorem}[Correctness of encoding from free-choice workflow net to CCS]\label{thm:free-wf-full}
    Let $N = (P, T, F, A, \sigma)$ be a free-choice workflow net and $M_0: P \to \mathbb{N}$ be a marking. $N$ and $M_0$ can be encoded into a weakly bisimilar CCS process $Q$ with defining equations $\mathcal{D}$ such that $LTS(N, M_0)$ has a divergent path iff $LTS(Q, \mathcal{D})$ has. The encoding time and the size of $(Q, \mathcal{D})$ are $\O(|N| + \sum_{p \in P} M_0(p))$.
\end{theorem}
\begin{proof}
    By \autoref{lem:free-wf-stepwise-invariant}+\ref{lem:free-wf-full-correctness}, \autoref{thm:free-wf-full-weak-bisimulation}+\ref{thm:ccs-net-strong-bisimulation} and transitivity of bisimulation~\cite{lts-book}.
    \qed
\end{proof}

\subsection{Encoding Any Free-Choice Net into CCS}\label{sec:free}

Although the focus of \autoref{sec:free-wf} is encoding free-choice \emph{workflow} nets into CCS, many definitions and results therein can be applied to \emph{any} free-choice net (\autoref{alg:free-wf-stepwise}, \autoref{thm:free-wf-stepwise-weak-bisimulation}, \autoref{lem:free-wf-stepwise-invariant}, \autoref{alg:free-wf-full}, \autoref{thm:free-wf-full-weak-bisimulation}). In this section we build upon them to develop an encoding from \emph{any} free-choice net into CCS.

To achieve this result, we define a superclass of CCS nets (\autoref{def:ccs-net}) called \emph{2-$\tau$-synchronisation nets} (\autoref{def:2-tau-synchronisation-net}, below), which may have transitions with no ingoing edges. Such transitions can be seen as \emph{token generators}: they can always fire and produce tokens; they do not occur in workflow nets, and they would be dropped if given to \autoref{alg:ccs-net} because they do not originate from any place. For this reason, token generator transitions are handled separately in \autoref{alg:free} (which extends \autoref{alg:ccs-net}).

\begin{definition}[2-$\tau$-synchronisation net]\label{def:2-tau-synchronisation-net}
    A Petri net $(P, T, F, A, \sigma)$ is a \emph{2-$\tau$-synchronisation net} iff each transition $t \in T$ has at most two ingoing edges --- and in the latter case, $\sigma(t) = \tau$.
\end{definition}

\begin{algorithm}[t]
    \DontPrintSemicolon
    \SetKwInOut{Input}{Input}
    \SetKwInOut{Output}{Output}
    \SetKw{Where}{where}

    \Input{2-$\tau$-synchronisation net $N = (P, T, F, A, \sigma)$ and marking $M_0: P \to \mathbb{N}$}
    \Output{CCS process $Q$ and partial mapping of defining equations $\mathcal{D}$}

    $(\_, \mathcal{D}) \gets \text{\autoref{alg:ccs-net}}(N, M_0)$\;
    $T_0 \gets t \in \left\{t \;\middle|\; t \in T \text{ and } t \text{ has 0 ingoing edges}\right\}$\;
    \For{$t \in T_0$}{
        $\mathcal{D}(X_t) \gets \sigma(t).(X_t \mathbin{|} X_{p_1} \mathbin{|} \dots \mathbin{|} X_{p_k})$ \Where $\{p_1, \dots, p_k\} = \{p \,|\, (t, p) \in F\}$\;
    }
    $Q \gets (\nu s_1) \dots (\nu s_n) \left({X_{p_1}^{M_0(p_1)}} \mathbin{|} \dots \mathbin{|} {X_{p_{|P|}}^{M_0(p_{|P|})}} \mathbin{|} \smash{\underbrace{X_{t_1} \mathbin{|} \dots \mathbin{|} X_{t_k}}_{t_i \in \{t_1, \dots, t_k\} = T_0}}\right)$\;
    \Return{$\left(Q,\; \mathcal{D}\right)$}

    \caption{Encoding from 2-$\tau$-synchronisation net to CCS process}
    \label{alg:free}
\end{algorithm}

Line 1 of \autoref{alg:free} retrieves the defining equations $\mathcal{D}$ returned by \autoref{alg:ccs-net} (the returned CCS process is not used). Line 2 defines $T_0$ as all transitions $t \in T$ with no ingoing edges. Each transition $t \in T_0$ is encoded as a sequential process named $X_t$ in $\mathcal{D}$ (line 3--5). $X_t$ is defined as $\sigma(t)$ followed by the \emph{parallel composition} of $X_t$ (itself) and all processes named $X_{p_i}$, where $p_i$ has an ingoing edge from $t$. This ensures that whenever $X_t$ is used, then it spawns a new copy of itself that can be used again. Finally on line 6, all $X_{t_i}$, where $t_i \in T_0$, are included once in the process $Q$ to be available from the start. For instance, the simple 2-$\tau$-synchronisation net $N = (\{p_1\}, \{t_1\}, \{(t_1, p_1)\}, \{b\}, [t_1 \mapsto b]\})$ and marking $M_0$ with $M_0(p_1) = 0$ are encoded into $(X_{t_1}, \mathcal{D})$ where $\mathcal{D}(X_{t_1}) = b.(X_{t_1} \mathbin{|} X_{p_1})$ which gives a process where $X_{t_1}$ is always present: $\smash{X_{t_1} \! \xrightarrow{b} \! (X_{p_1} \mathbin{|} X_{t_1}) \! \xrightarrow{b} \! (X_{p_1}^2 \mathbin{|} X_{t_1}) \! \xrightarrow{b} \! \dots}$
\autoref{alg:free} also produces strongly bisimilar CCS processes like \autoref{alg:ccs-net}.

\begin{theorem}[Correctness of \autoref{alg:free}]\label{thm:free-strong-bisimulation}
    Given a 2-$\tau$-synchronisation net $N = (P, T, F, A, \sigma)$ and an initial marking $M_0: P \to \mathbb{N}$, the result of applying \autoref{alg:free} on $N$ and $M_0$ is $(Q, \mathcal{D})$ such that $LTS(N, M_0) \bisimilar LTS(Q, \mathcal{D})$. The translation time and the size of $(Q, \mathcal{D})$ are bound by $\O(|N| + \sum_{p \in P} M_0(p))$.
\end{theorem}
\begin{proof}
    \emph{(Sketch, detailed proof in \appendixref{sec:appendix-free-strong-bisimulation}.)} Extend the proof of \autoref{thm:ccs-net-strong-bisimulation} by extending the parallel composition in \emph{all} CCS processes in the relation $\mathcal{R}$ with the parallel composition of all $X_{t_i}$, where $t_i \in T_0$.
    \qed
\end{proof}

The CCS process obtained from \autoref{alg:free} is linear in size of the encoded 2-$\tau$-synchronisation net and marking. \autoref{thm:free-strong-bisimulation} allows us to extend the encoding of free-choice workflow nets to free-choice nets.

\begin{lemma}[\autoref{alg:free-wf-full} output]\label{lem:free-correctness}
    If applied to a free-choice net $(P, T, F, A, \sigma)$, \autoref{alg:free-wf-full} returns a 2-$\tau$-synchronisation net.
    \emph{(Proof: \appendixref{sec:appendix-free-correctness}.)}
\end{lemma}

\begin{theorem}[Correctness of free-choice net to CCS encoding]\label{thm:free}
    Let $N = (P, T, F, A, \sigma)$ be a free-choice net and $M_0: P \to \mathbb{N}$ be an initial marking. $N$ and $M_0$ can be encoded into a weakly bisimilar CCS process $Q$ with defining equations $\mathcal{D}$ such that $LTS(N, M_0)$ has a divergent path iff $LTS(Q, \mathcal{D})$ has. The encoding time and the size of $(Q, \mathcal{D})$ are $\O(|N| + \sum_{p \in P} M_0(p))$.
\end{theorem}
\begin{proof}
    By \autoref{lem:free-wf-stepwise-invariant}+\ref{lem:free-correctness}, \autoref{thm:free-wf-full-weak-bisimulation}+\ref{thm:free-strong-bisimulation} and transitivity of bisimulation~\cite{lts-book}.
    \qed
\end{proof}

\subsection{Group-Choice Nets}\label{sec:group}

We now further extend our CCS encodability results. Free-choice nets (\autoref{def:free-choice-net}) only allow for choices taken by a \emph{single} place; we relax this requirement by defining a larger class of nets called \emph{group-choice nets} (\autoref{def:group-choice-net}, below), where choices can be taken by a \emph{group} of places which synchronise before making the choice. This requires that all places in the group \emph{agree} on their options --- which corresponds to all places in the group having the same postset (i.e., edges to the same transitions). No other places outside the group may have edges to those transitions (as that would affect the choice).

\begin{definition}[Postset of a place]\label{def:transition-set}
    Let $(P, T, F, A, \sigma)$ be a Petri net. The \emph{postset} of a place $p \in P$ (written ${p\bullet}$) is the set of all transitions with an ingoing edge from $p$, i.e.: ${p\bullet} \mathrel{\coloneq} \{t \,|\, p \in P, (p, t) \in F\}$.
\end{definition}

\begin{definition}[Group-choice net]\label{def:group-choice-net}
    A Petri net $N$ is a \emph{group-choice net} iff all pairs of places either have a same postset, or disjoint postsets.
\end{definition}

\begin{figure}[!b]
    \centering
    \begin{minipage}[b][][b]{0.28\textwidth}
        \centering
        \includegraphics[scale=0.30]{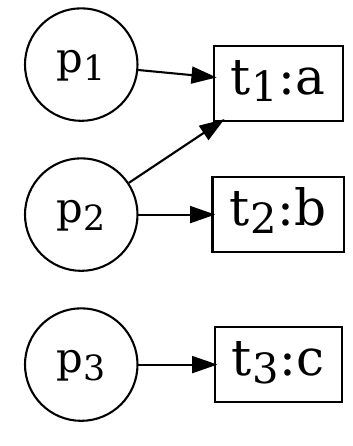}
        \caption{Petri net with $p_1\bullet = \{t_1\}$, $p_2\bullet = \{t_1, t_2\}$ and $p_3\bullet = \{t_3\}$.}
        \label{fig:group-not}
    \end{minipage}\hfill%
    \begin{minipage}[b][][b]{0.28\textwidth}
        \centering
        \includegraphics[scale=0.30]{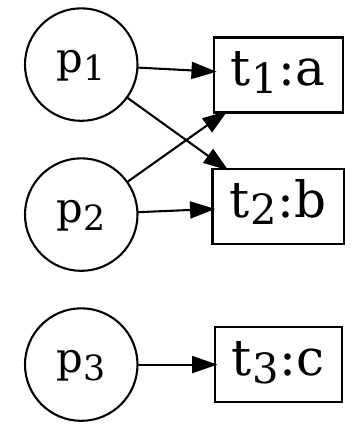}
        \caption{Group-choice \linebreak[4] net with $p_1\bullet = p_2\bullet = \{t_1, t_2\}$ and $p_3\bullet = \{t_3\}$.}
        \label{fig:group-is}
    \end{minipage}\hfill%
    \begin{minipage}[b][][b]{0.38\textwidth}
        \centering
        \includegraphics[scale=0.30]{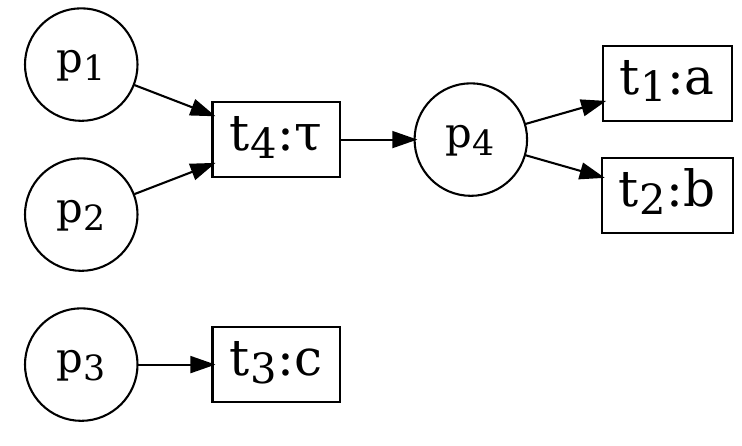}
        \caption{Result of \autoref{alg:group-stepwise} on \autoref{fig:group-is}. $p_1\bullet = p_2\bullet = \{t_4\}$, $p_3\bullet \mathbin{=} \{t_3\}$ and $p_4\bullet \mathbin{=} \{t_1, t_2\}$.}
        \label{fig:group-transformed}
    \end{minipage}
\end{figure}

A group-choice net does not have partially overlapping place postsets. Therefore, its places can be split into \emph{disjoint groups} where all places in a group have the same postset. \autoref{fig:group-not}--\ref{fig:group-transformed} show examples of Petri nets and postsets for each place. \autoref{fig:group-not} is not a group-choice net because ${p_1\bullet}$ and ${p_2\bullet}$ are only partially overlapping (since $p_2$ has an edge to $t_2$ while $p_1$ does not). \autoref{fig:group-is} is a group-choice net, hence its places can be split into two disjoint groups $\{p_1, p_2\}$ and $\{p_3\}$. \autoref{fig:group-transformed} is a group-choice net with three groups of places.

A place postset of size at least two (like $p_4\bullet$ in \autoref{fig:group-transformed}) corresponds to a choice because the place can choose between at least two transitions. When at least two places have the same postset (like $p_1\bullet$ and $p_2\bullet$ in \autoref{fig:group-transformed}), it corresponds to synchronisation because the transitions in the postset have at least two ingoing edges. The definition of group-choice nets ensures that choice and synchronisation can only happen at the same time if the group of places can safely synchronise before making a choice together. A free-choice net further require that every place postset must not correspond to both choice and synchronisation.

\begin{algorithm}[t]
    \DontPrintSemicolon
    \SetKwInOut{Input}{Input}
    \SetKwInOut{Output}{Output}
    \SetKw{Where}{where}

    \Input{Petri net $N = (P, T, F, A, \sigma)$, marking $M_0: P \to \mathbb{N}$ and two selected places $p^*, p^{**} \in P$ where $p^*\bullet = p^{**}\bullet \neq \emptyset$}
    \Output{Petri net $(P', T', F', A', \sigma')$ and marking $M_0': P' \to \mathbb{N}$}

    $P' \gets P \cup \{p^+\}$ \Where $p^+ \notin P \cup T$\;
    $T' \gets T \cup \{t^+\}$ \Where $t^+ \notin P' \cup T$\;
    $F' \gets \left(
      \begin{array}{@{}l@{}}
        \Big(F \setminus \big(\{(p^*, t) \,|\, t \in p^*\bullet\} \cup \{(p^{**}, t) \,|\, t \in p^{**}\bullet\}\big)\Big)
        \\
        \quad\phantom{X} \;\cup\; \{(p^*, t^+), (p^{**}, t^+), (t^+, p^+)\} \;\cup\; \{(p^+, t) \,|\, t \in p^*\bullet\}
      \end{array}\right)$\;
    $N' \gets (P', T', F', A \cup \{\tau\}, \sigma[t^+ \mapsto \tau])$\;
    \Return{$(N', M_0[p^+ \mapsto 0])$}\;

    \caption{Group-choice net transition preset reduction}
    \label{alg:group-stepwise}
\end{algorithm}

\begin{algorithm}[b]
    \DontPrintSemicolon
    \SetKwInOut{Input}{Input}
    \SetKwInOut{Output}{Output}
    \SetKw{Select}{select}
    \SetKw{Where}{where}
    \SetKw{And}{and}
    \SetKw{Or}{or}

    \Input{Group-choice net $N = (P, T, F, A, \sigma)$ and marking $M_0: P \to \mathbb{N}$}
    \Output{Petri net $(P', T', F', A', \sigma')$ and marking $M_0': P' \to \mathbb{N}$}

    $(N', M_0') \gets (N, M_0)$\;
    \While{$\exists t^* \in T'$ \Where $|\{p \,|\, (p, t^*) \in F\}| > (\text{\leIf*{$\sigma(t^*) = \tau$}{$2$}{$1$}})$}{
        \Select $(p_1, t^*), (p_2, t^*) \in F$ \Where $p_1 \neq p_2$\;
        $(N', M_0') \gets \text{\autoref{alg:group-stepwise}}(N', M_0', p_1, p_2)$\;
    }
    \Return{$(N', M_0')$}\;

    \caption{Iterative group-choice net transition preset reduction}
    \label{alg:group-full}
\end{algorithm}

The transformation of group-choice nets into 2-$\tau$-synchronisation nets is intuitively similar to the transformation of free-choice nets (in \autoref{sec:free-wf}). The key difference is that, since places in group-choice nets can have more than one outgoing edge when they synchronise on a transition, we consider whole place postsets instead of a single transition. \autoref{alg:group-stepwise} does this by selecting two places $p^*$ and $p^{**}$ with the same (non-empty) postset: It removes all their outgoing edges (line 3) and connects them to a new $\tau$-transition $t^+$ (defined on line 2) with an edge to a new place $p^+$ (defined on line 1) that has outgoing edges to the same postset $p^*$ and $p^{**}$ originally had (lines 3--5). This reduces the number of ingoing edges (i.e., the \emph{preset}) of \emph{all} transitions in $p^*\bullet$ by one. \autoref{fig:group-transformed} shows the result of one application of \autoref{alg:group-stepwise} on \autoref{fig:group-is} with $p^* = p_1$ and $p^{**} = p_2$.

\autoref{alg:group-full} keeps applying \autoref{alg:group-stepwise} on the input group-choice net (in the same fashion as \autoref{alg:free-wf-full}) by non-deterministically picking $p_1$ and $p_2$, until a 2-$\tau$-synchronisation net is obtained. All the proofs for the encoding of group-choice nets are similar to the ones for free-choice nets: The only difference is that the proofs for group-choice nets consider \emph{sets} of impacted transitions instead of single transitions. Therefore, we only present the main result below.

\begin{theorem}[Correctness of group-choice net to CCS encoding]\label{thm:group}
    A group-choice net $N = (P, T, F, A, \sigma)$ and an initial marking $M_0: P \to \mathbb{N}$ can be encoded into a weakly bisimilar CCS process $Q$ with defining equations $\mathcal{D}$ s.t. $LTS(N, M_0)$ has a divergent path iff $LTS(Q, \mathcal{D})$ has. The encoding time and the size of $(Q, \mathcal{D})$ are $\O(|N| + \sum_{p \in P} M_0(p))$. \emph{(Proofs in \appendixref{sec:appendix-group-sec}.)}
\end{theorem}

\section{Conclusion and Future Work}\label{sec:conclusion-future}

In this paper we have formalised a direct encoding of CCS nets into strongly bisimilar CCS processes (akin to~\cite{multi-ccs}), and on this foundation we have developed novel Petri net-to-CCS encodings and results, through a series of generalizations. We have introduced an encoding of \emph{free-choice workflow nets} into weakly bisimilar CCS processes. We have generalized this result as an encoding from \emph{any} free-choice net into a weakly bisimilar CCS processes, via a new superclass of CCS nets called \emph{2-$\tau$-synchronisation nets} (which may have token-generating transitions). Finally, we have presented a superclass of free-choice nets called \emph{group-choice nets}, and its encoding into weakly bisimilar CCS processes.

On the practical side, we are now exploring the practical application of these results --- e.g.~using model checkers oriented towards process calculi (such as mCRL2~\cite{DBLP:conf/tacas/BunteGKLNVWWW19}) to analyse the properties of $\alpha$-mined Petri nets encoded into CCS.
On the theoretical side, we are studying how to encode even larger classes of Petri nets into weakly bisimilar CCS processes. One of the problems is how to handle partially-overlapping place postsets: in this paper, we focused on classes of Petri nets without such overlaps, and 2-$\tau$ synchronisation nets. Another avenue we are exploring is to develop further encodings after dividing a Petri net into groups of places/transitions, as we did for group-choice nets.

\subsubsection*{Acknowledgements.}
We wish to thank to the anonymous reviewers for their comments and suggestions, and Ekkart Kindler for the fruitful discussions.

\bibliographystyle{splncs04}
\bibliography{paper}

\clearpage
\appendix
\section{Extended Preliminaries}\label{sec:appendix-definitions}

\begin{definition}[Weak closure of $\to$~{\cite[Definition 2.16]{lts-book}}]\label{def:appendix-weak-closure}
    For $TS = (Q, A, \to)$, define the relation $\implies \subseteq Q \times A^{\prime*} \times Q$ as the least closure of $\to$, where $A^{\prime*}$ is a trace of visible actions from $A^\prime \mathrel{\coloneqq} A \setminus \{\tau\}$, induced by ($\epsilon$ is the empty trace):
    \begin{align*}
        \frac{q_1 \overset{a}{\longrightarrow} q_2}{q_1 \overset{a}{\implies} q_2} && \frac{q_1 \overset{\tau}{\longrightarrow} q_2}{q_1 \overset{\epsilon}{\implies} q_2} && \frac{}{q \overset{\epsilon}{\implies} q} && \frac{q_1 \overset{\sigma_1}{\implies} q_2 \quad q_2 \overset{\sigma_2}{\implies} q_3}{q_1 \overset{\sigma_1 \sigma_2}{\implies} q_3}
    \end{align*}
\end{definition}

\begin{definition}[Finite-net CCS~{\cite[Section 3.4]{lts-book}}]\label{def:appendix-finite-net-ccs}
    A CCS process $Q$ and its partial mapping of defining equations $\mathcal{D}$ is \emph{finite-net CCS} when no definitions in $\mathcal{D}$ contains restrictions.
\end{definition}

\begin{restate}{definition}{def:free-choice-net}
\begin{definition}[Free choice net~{\cite{free-choice-nets}} - extended]\label{def:appendix-free-choice-net}
    A Petri net $(P, T, F, A, \sigma)$ is a \emph{free-choice net} iff it satisfies the following two properties:
    \begin{itemize}
        \lessspace
        \item \textbf{Unique choice}: All places $p \in P$ with more than one outgoing edge only have edges to transitions with exactly one ingoing edge (edge from $p$):
        \begin{equation}\label{eq:appendix-free-wf-unique-choice}
            \forall p. \forall t. p \in P \land (p, t) \in F \land |\{t^\prime | (p, t^\prime) \in F\}| > 1 \implies |\{p^\prime | (p^\prime, t) \in F\}| = 1
        \end{equation}
        \item \textbf{Unique synchronisation}: All transitions $t \in T$ with more than one ingoing edge only have edges from places with exactly one outgoing edge (edge to $t$):
        \begin{equation}\label{eq:appendix-free-wf-unique-sync}
            \forall p. \forall t. t \in T \land (p, t) \in F \land |\{p^\prime | (p^\prime, t) \in F\}| > 1 \implies |\{t^\prime | (p, t^\prime) \in F\}| = 1
        \end{equation}
        \lessspace
    \end{itemize}
\end{definition}
\label{sec:appendix-free-wf-choice-net}
\end{restate}

\begin{lemma}\label{lem:appendix-free-wf-choice-implies-sync}
    Let $N = (P, T, F, A, \sigma)$ be a Petri Net. If $N$ satisfies \emph{unique choice} in \autoref{def:free-choice-net} then it also satisfies \emph{unique synchronisation} in \autoref{def:free-choice-net}.
\end{lemma}
\begin{proof}
    Assume that unique choice \eqref{eq:appendix-free-wf-unique-choice} holds:
    \begin{equation*}
        \forall p. \forall t. p \in P \land (p, t) \in F \land |\{t^\prime | (p, t^\prime) \in F\}| > 1 \implies |\{p^\prime | (p^\prime, t) \in F\}| = 1
    \end{equation*}
    Consider an arbitrary transition $t$ with more than one ingoing edge ($|\{p^\prime | (p^\prime, t) \in F\}| > 1$). Consider an arbitrary place $p$ with an outgoing edge to $t$ i.e. $p \in \{p^\prime | (p^\prime, t) \in F\}$. The place $p$ has at least one outgoing edge namely the one to $t$. Assume that $p$ has an edge to another transition $t^\prime$ ($t^\prime \neq t$). Then $p$ has two outgoing edges such that $|\{t^\prime | (p, t^\prime) \in F\}| > 1$ which implies $|\{p^\prime | (p^\prime, t) \in F\}| = 1$ according to the initial assumption. However, this is a contradiction since $t$ was assumed to have more than one ingoing edge ($|\{p^\prime | (p^\prime, t) \in F\}| > 1$). Hence, every place $p$ with an edge to $t$ must have exactly one outgoing edge ($|\{t^\prime | (p, t^\prime) \in F\}| = 1$), namely the one to $t$. Hence, this shows $|\{p^\prime | (p^\prime, t) \in F\}| > 1 \implies |\{t^\prime | (p, t^\prime) \in F\}| = 1$. Since $t$ and $p$ was chosen arbitrary, this is true for all $(p, t) \in F$.
    \qed
\end{proof}

\begin{lemma}\label{lem:appendix-free-wf-sync-implies-choice}
    Let $N = (P, T, F, A, \sigma)$ be a Petri Net. If $N$ satisfies \emph{unique synchronisation} in \autoref{def:free-choice-net} then it also satisfies \emph{unique choice} in \autoref{def:free-choice-net}.
\end{lemma}
\begin{proof}
    Assume that unique synchronisation \eqref{eq:appendix-free-wf-unique-sync} holds:
    \begin{equation*}
        \forall p. \forall t. p \in P \land (p, t) \in F \land |\{t^\prime | (p, t^\prime) \in F\}| > 1 \implies |\{p^\prime | (p^\prime, t) \in F\}| = 1
    \end{equation*}
    Consider an arbitrary place $p$ with more than one outgoing edge ($|\{t^\prime | (p, t^\prime) \in F\}| > 1$). Consider an arbitrary transition $t$ with an ingoing edge from $t$ i.e. $t \in \{t^\prime | (p, t^\prime) \in F\}$. The transition $t$ has at least one ingoing edge namely the one from $p$. Assume that $t$ has an edge from another place $p^\prime$ ($p^\prime \neq p$). Then $t$ has two ingoing edges such that $|\{p^\prime | (p^\prime, t) \in F\}| > 1$ which implies $|\{t^\prime | (p, t^\prime) \in F\}| = 1$ according to the initial assumption. However, this is a contradiction since $p$ was assumed to have more than one outgoing edge ($|\{t^\prime | (redp, t^\prime) \in F\}| > 1$). Hence, every transition $t$ with an edge from $p$ must have exactly one ingoing edge ($|\{p^\prime | (p^\prime, t) \in F\}| = 1$), namely the one from $p$. Hence, this shows $|\{t^\prime | (p, t^\prime) \in F\}| > 1 \implies |\{p^\prime | (p^\prime, t) \in F\}| = 1$. Since $p$ and $t$ was chosen arbitrary, this is true for all $(p, t) \in F$.
    \qed
\end{proof}

\begin{theorem}\label{lem:appendix-free-wf-choice-iff-sync}
    Let $N = (P, T, F, A, \sigma)$ be a Petri Net. $N$ satisfies \emph{unique choice} in \autoref{def:free-choice-net} iff it satisfies \emph{unique synchronisation} in \autoref{def:free-choice-net}.
\end{theorem}
\begin{proof}
    Follows from \autoref{lem:appendix-free-wf-choice-implies-sync} and \autoref{lem:appendix-free-wf-sync-implies-choice}.
    \qed
\end{proof}

\section{Proofs for Encoding CCS Nets into CCS Processes}\label{sec:appendix-ccs-net}

\begin{lemma}[\autoref{alg:ccs-net} output]\label{lem:appendix-ccs-net-finite-net-ccs}
   If applied to a CCS net $(P, T, F, A, \sigma)$, \autoref{alg:ccs-net} returns valid finite-net CCS.
\end{lemma}
\begin{proof}
    The output of \autoref{alg:ccs-net} is a CCS process $Q$ and a partial mapping of defining equations $\mathcal{D}$. For all $p \in P$, $\mathcal{D}(X_p)$ is the choice of substitutions for $Y_t$. All substitutions are sequential processes. Clearly, $Q$ is a valid process and thereby \autoref{alg:ccs-net} return valid CCS. There are no restrictions in $\mathcal{D}(X_p)$ for all $p \in P$, so it is also finite-net CCS.
    \qed
\end{proof}

\begin{restate}{theorem}{thm:ccs-net-strong-bisimulation}
\begin{theorem}[Correctness of \autoref{alg:ccs-net}]\label{thm:appendix-ccs-net-strong-bisimulation}
    Given a CCS net $N = (P, T, F, A, \sigma)$ and an initial marking $M_0$, let the result of applying \autoref{alg:ccs-net} to $N$ and $M_0$ be the CCS process $Q_0$ and defining equations $\mathcal{D}$. Then, $LTS(N, M_0) \bisimilar LTS(Q_0, \mathcal{D})$. The translation time and the size of $(Q_0, \mathcal{D})$ are bound by $\O(|N| + \sum_{p \in P} M_0(p))$. ($Q_0$ is used here instead of $Q$ to avoid confusions in the proof.)
\end{theorem}
\label{sec:appendix-ccs-net-strong-bisimulation}
\begin{proof}
    Markings are used as the states in the Petri net while the parallel composition of process names is used as the states in the CCS process. Define the bisimulation relation $\mathcal{R}$ as follows:
    \begin{align*}
        \mathcal{R} \mathrel{\;\coloneqq\;} \left\{\left(M,\ (\nu s_1)\dots(\nu s_n)\left({X_{p_1}^{M(p_1)}} \mathbin{|} \dots \mathbin{|} {X_{p_{|P|}}^{M(p_{|P|})}}\right)\right) \;\middle|\; M \in LTS(N, M_0)\right\}
    \end{align*}
    where $M$ is a reachable marking from the initial marking $M_0$ in $N$ (i.e.~$M$ is a state in $LTS(N, M_0)$), and $s_1, \dots, s_n$ are the restricted names in $Q$. The intuition is as that for every place $p$ containing $M(p)$ tokens, the CCS process contains $M(p)$ replicas of the process named $X_p$ (written $X_p^{M(p)}$). The relation can be seen as an invariant that should hold after each step. The restrictions are left out in the rest of the proof since they are static.

    \paragraph{Case 0}
    Clearly, the initial marking $M_0$ and the initial process $Q_0$ are in the relation $(M_0, Q_0) \in \mathcal{R}$ because $Q_0$ is defined using $M_0$.

    \paragraph{Case 1}
    Consider any pair $(M, Q) \in \mathcal{R}$ i.e. a marking $M$ and its related process $Q$. Consider an arbitrary enabled transition $t \in T$ in $N$ with respect to $M$. Define all places with an edge from $t$ in $N$ as $P_{out} \mathrel{\coloneqq} \left\{p \;\middle|\; (t, p) \in F\right\}$ and let $k \mathrel{\coloneqq} |P_{out}|$. Without loss of generality, the places in $P_{out}$ are assumed to be enumerated as the first $k$ of the $|P|$ places i.e. $\{p_1, p_2, \dots, p_k\} = P_{out}$.

    \paragraph{Case 1-1}
    First, assume that $t \in T_1$. Since $t$ is enabled, the edge $(p^*, t) \in F$ gives that there is a token at $p^*$, formally $M(p^*) \geq 1$. This means that there is at least one active instance of the place process $X_{p^*}$ according to the relation. The edge $(p^*, t) \in F$ means that a replacement of $Y_t$ is one of the choices in $\mathcal{D}(X_{p^*})$. The replacement is $\sigma(t).(X_{p_1} \mathbin{|} \dots \mathbin{|} X_{p_k})$ because $t \in T_1$. Hence, $\sigma(t)$ can be executed in $Q$.

    Firing $t$ in $N$ emits a $\sigma(t)$, consumes one token from $p^*$ and produces one token to each place $p \in P_{out}$ which results in the marking $M^\prime$:
    \begin{align}
        \label{eq:appendix-ccs-net-strong-bisimulation-t1}
        M^\prime(p) \mathrel{\;\coloneqq\;} \begin{cases}
            M(p) - 1 & \text{if } p = p^* \land p \notin P_{out} \\
            M(p) + 1 & \text{if } p \neq p^* \land p \in P_{out} \\
            M(p) & \text{otherwise}
        \end{cases}
    \end{align}
    Executing $\sigma(t)$ in $Q$ results in $Q^\prime$ where one instance of $X_{p^*}$ is replaced with $(X_{p_1} \mathbin{|} \dots \mathbin{|} X_{p_k})$. This is the same as removing one instance of $X_{p^*}$ and adding one instance of $X_p$ for each $p \in P_{out}$. There are two cases, either $p^* \in P_{out}$ or $p^* \notin P_{out}$.

    \paragraph{Case 1-1-1}
    In the first case, assume without loss of generality that $p^* = p_1$ such that $Q^\prime$ is:
    \begin{align*}
        X_{p_1}^{M(p_1) - 1 + 1} \mathbin{|} X_{p_2}^{M(p_2) + 1} \mathbin{|} \dots \mathbin{|} X_{p_k}^{M(p_k) + 1} \mathbin{|} X_{p_{k + 1}}^{M(p_{k + 1})} \mathbin{|} \dots \mathbin{|} X_{p_{|P|}}^{M(p_{|P|})}
    \end{align*}
    In $N$, $p_1$ hits the third case, $p_2, \dots, p_k$ hits the second case and $p_{k + 1}, \dots, p_{|P|}$ hits the third case in \eqref{eq:appendix-ccs-net-strong-bisimulation-t1}. Hence, $Q^\prime$ can be rewritten to:
    \begin{align*}
        X_{p_1}^{M^\prime(p_1)} \mathbin{|} X_{p_2}^{M^\prime(p_2)} \mathbin{|} \dots \mathbin{|} X_{p_k}^{M^\prime(p_k)} \mathbin{|} X_{p_{k + 1}}^{M^\prime(p_{k + 1})} \mathbin{|} \dots \mathbin{|} X_{p_{|P|}}^{M^\prime(p_{|P|})}
    \end{align*}
    This shows that $(M^\prime, Q^\prime) \in \mathcal{R}$ holds for this case.

    \paragraph{Case 1-1-2}
    Otherwise, $p^* \notin P_{out}$ and without loss of generality it is assumed that $p^* = p_{k + 1}$ such that $Q^\prime$ is:
    \begin{align*}
        X_{p_1}^{M(p_1) + 1} \mathbin{|} \dots \mathbin{|} X_{p_k}^{M(p_k) + 1} \mathbin{|} X_{p_{k + 1}}^{M(p_{k + 1}) - 1} \mathbin{|} X_{p_{k + 2}}^{M(p_{k + 2})} \mathbin{|} \dots \mathbin{|} X_{p_{|P|}}^{M(p_{|P|})}
    \end{align*}
    $p_1, \dots, p_k$ hits the second case, $p_{k + 1}$ hits the first case and $p_{k + 1}, \dots, p_{|P|}$ hits the third case in \eqref{eq:appendix-ccs-net-strong-bisimulation-t1} such that $Q^\prime$ becomes:
    \begin{align*}
        X_{p_1}^{M^\prime(p_1)} \mathbin{|} \dots \mathbin{|} X_{p_k}^{M^\prime(p_k)} \mathbin{|} X_{p_{k + 1}}^{M^\prime(p_{k + 1})} \mathbin{|} X_{p_{k + 2}}^{M^\prime(p_{k + 2})} \mathbin{|} \dots \mathbin{|} X_{p_{|P|}}^{M^\prime(p_{|P|})}
    \end{align*}
    This shows that $(M^\prime, Q^\prime) \in \mathcal{R}$ holds for $t \in T_1$.

    \paragraph{Case 1-2}
    Now, assume that $t \in T_2$ and thus $t$ is a $\tau$-transition. In this case, there are two ingoing edges $(p^*, t), (p^{**}, t) \in F$ where $p^* \neq p^{**}$ meaning that $M(p^*) \geq 1$ and $M(p^{**}) \geq 1$. Hence, there are at least one active instance of both the place processes $X_{p^*}$ and $X_{p^{**}}$. The edges $(p^*, t) \in F$ and $(p^{**}, t) \in F$ means that a replacement of $T_t$ is one of the choices in both $\mathcal{D}(X_{p^*})$ and $\mathcal{D}(X_{p^{**}})$. The replacements must be different since the case $t \in T_2$ results in two different replacements, namely $s_t.(X_{p_1} \mathbin{|} \dots \mathbin{|} X_{p_k})$ and $\overline{s_t}.\mathbf{0}$. Hence, one of the replacements is a choice in $\mathcal{D}(X_{p^*})$ while the other replacement is a choice in $\mathcal{D}(X_{p^{**}})$. This means that $s_t$ and $\overline{s_t}$ can \emph{synchronize} which results in a $\tau$-action.

    Firing $t$ in $N$ emits a $\tau$ and results in the marking $M^\prime$:
    \begin{align}
        \label{eq:appendix-ccs-net-strong-bisimulation-t2}
        M^\prime(p) \mathrel{\;\coloneqq\;} \begin{cases}
            M(p) - 1 & \text{if } p \in \{p^*, p^{**}\} \land p \notin P_{out} \\
            M(p) + 1 & \text{if } p \notin \{p^*, p^{**}\} \land p \in P_{out} \\
            M(p) & \text{otherwise}
        \end{cases}
    \end{align}
    Letting $s_t$ and $\overline{s_t}$ react in $Q$ emits a $\tau$ and results in $Q^\prime$ where $(X_{p^*} \mathbin{|} X_{p^{**}})$ is replaced with $((X_{p_1} \mathbin{|} \dots \mathbin{|} X_{p_k}) \mathbin{|} \mathbf{0})$. This is the same as removing one $X_{p^*}$, removing one $X_{p^{**}}$ and adding one $X_p$ for each $p \in P_{out}$. There are four cases:

    \paragraph{Case 1-2-1}
    The first case is $p^*, p^{**} \in P_{out}$ such that it without loss of generality can be assumed that $p^* = p_1$ and $p^{**} = p_2$. Hence, $Q^\prime$ is:
    \begin{align*}
        X_{p_1}^{M(p_1) - 1 + 1} \mathbin{|} X_{p_2}^{M(p_2) - 1 + 1} \mathbin{|} X_{p_3}^{M(p_3) + 1} \mathbin{|} \dots \mathbin{|} X_{p_k}^{M(p_k) + 1} \mathbin{|} X_{p_{k + 1}}^{M(p_{k + 1})} \mathbin{|} \dots \mathbin{|} X_{p_{|P|}}^{M(p_{|P|})}
    \end{align*}
    $p_1, p_2, p_{k + 1}, \dots, p_{|P|}$ hits the third case and $p_3, \dots, p_k$ hits the second case in \eqref{eq:appendix-ccs-net-strong-bisimulation-t2} which means that $Q^\prime$ can be rewritten to:
    \begin{align*}
        X_{p_1}^{M^\prime(p_1)} \mathbin{|} X_{p_2}^{M^\prime(p_2)} \mathbin{|} X_{p_3}^{M^\prime(p_3)} \mathbin{|} \dots \mathbin{|} X_{p_k}^{M^\prime(p_k)} \mathbin{|} X_{p_{k + 1}}^{M^\prime(p_{k + 1})} \mathbin{|} \dots \mathbin{|} X_{p_{|P|}}^{M^\prime(p_{|P|})}
    \end{align*}
    Hence, $(M^\prime, Q^\prime) \in \mathcal{R}$ in this case.

    \paragraph{Case 1-2-2}
    The second case is $p^* \in P_{out}$ and $p^* \notin P_{out}$. Without loss of generality, it is assumed that $p^* = p_1$ and $p^{**} = p_{k + 1}$. Hence, $Q^\prime$ is:
    \begin{align*}
        X_{p_1}^{M(p_1) - 1 + 1} \mathbin{|} X_{p_2}^{M(p_2) + 1} \mathbin{|} \dots \mathbin{|} X_{p_k}^{M(p_k) + 1} \mathbin{|} X_{p_{k + 1}}^{M(p_{k + 1}) - 1} \mathbin{|} X_{p_{k + 2}}^{M(p_{k + 2})} \mathbin{|} \dots \mathbin{|} X_{p_{|P|}}^{M(p_{|P|})}
    \end{align*}
    $p_1, p_{k + 2}, \dots, p_{|P|}$ hits the third case, $p_2, \dots, p_k$ hits the second case and $p_{k + 1}$ hits the first case in \eqref{eq:appendix-ccs-net-strong-bisimulation-t2} which means that $Q^\prime$ can be rewritten to:
    \begin{align*}
        X_{p_1}^{M^\prime(p_1)} \mathbin{|} X_{p_2}^{M^\prime(p_2)} \mathbin{|} \dots \mathbin{|} X_{p_k}^{M^\prime(p_k)} \mathbin{|} X_{p_{k + 1}}^{M^\prime(p_{k + 1})} \mathbin{|} X_{p_{k + 2}}^{M^\prime(p_{k + 2})} \mathbin{|} \dots \mathbin{|} X_{p_{|P|}}^{M^\prime(p_{|P|})}
    \end{align*}
    Hence, $(M^\prime, Q^\prime) \in \mathcal{R}$ in this case.

    \paragraph{Case 1-2-3}
    The third case is $p^* \notin P_{out}$ and $p^* \in P_{out}$. Without loss of generality, it can be assumed that $p^* = p_{k + 1}$ and $p^{**} = p_1$. The rest of the proof for this case is the same as \emph{Case 1-2-2}.

    \paragraph{Case 1-2-4}
    The last case is $p^*, p^{**} \notin P_{out}$ such that it without loss of generality can be assumed that $p^* = p_{k + 1}$ and $p^{**} = p_{k + 2}$. Hence, $Q^\prime$ is:
    \begin{align*}
        X_{p_1}^{M(p_1) + 1} \mathbin{|} \dots \mathbin{|} X_{p_k}^{M(p_k) + 1} \mathbin{|} X_{p_{k + 1}}^{M(p_{k + 1}) - 1} \mathbin{|} X_{p_{k + 2}}^{M(p_{k + 2}) - 1} \mathbin{|} X_{p_{k + 3}}^{M(p_{k + 3})} \mathbin{|} \dots \mathbin{|} X_{p_{|P|}}^{M(p_{|P|})}
    \end{align*}
    $p_1, \dots, p_k$ hits the second case, $p_{k + 1}$ and $p_{k + 2}$ hits the first case while $p_{k + 3}, \dots, p_{|P|}$ hits the third case in \eqref{eq:appendix-ccs-net-strong-bisimulation-t2} which means that $Q^\prime$ can be rewritten to:
    \begin{align*}
        X_{p_1}^{M^\prime(p_1)} \mathbin{|} \dots \mathbin{|} X_{p_k}^{M^\prime(p_k)} \mathbin{|} X_{p_{k + 1}}^{M^\prime(p_{k + 1})} \mathbin{|} X_{p_{k + 2}}^{M^\prime(p_{k + 2})} \mathbin{|} X_{p_{k + 3}}^{M^\prime(p_{k + 3})} \mathbin{|} \dots \mathbin{|} X_{p_{|P|}}^{M^\prime(p_{|P|})}
    \end{align*}
    This shows that $(M^\prime, Q^\prime) \in \mathcal{R}$ holds for $t \in T_2$. Hence, the complete proof of \emph{Case 1} shows that $LTS(Q_0, \mathcal{D})$ strongly simulates $LTS(N, M_0)$ using $\mathcal{R}$.

    \paragraph{Case 2}
    Consider $(M, Q) \in \mathcal{R}$ i.e. a process $Q$ and its related marking $M$.

    \paragraph{Case 2-1}
    Consider a non-restricted action $a$ in $Q$ that can be executed. The only non-restricted actions created by \autoref{alg:ccs-net} are generated for transitions with one ingoing edge, namely for $t \in T_1$. Hence, $a = \sigma(t)$ must be the case. $Q$ is the parallel composition of place processes. Each place process $X_p$ represents a place $p \in P$. $\mathcal{D}(X_p)$ is the choice between \emph{replacements} of transition processes $Y_{t_1} + Y_{t_2} + \dots + Y_{t_l}$ representing transitions $t_i \in \{t_1, t_2, \dots, t_l\} = \left\{t \;\middle|\; (p, t) \in F\right\}$ that has an ingoing edge from $p$. Since $t \in T_1$ then there exists exactly one place $p^* \in P$ such that $(p^*, t) \in F$. This means that one of the choices in $\mathcal{D}(X_{p^*})$ is $\sigma(t).(X_{p_1} \mathbin{|} X_{p_2} \mathbin{|} \dots \mathbin{|} X_{p_k})$ where $p_i \in \{p_1, p_2, \dots, p_k\} = \left\{p \;\middle|\; (t, p) \in F\right\}$. Since $\sigma(t)$ can be executed in $Q$ then $Q$ must contain at least one instance of $X_{p^*}$ which by the bisimulation relation $\mathcal{R}$ means that $M(p^*) \geq 1$. $t \in T_1$ implies that $(p^*, t) \in F$ is the only ingoing edge to $t$ and thus $t$ can be fired in $N$.

    The rest of the proof for this case is the same as \emph{Case 1-1}.

    \paragraph{Case 2-2}
    Consider a restricted action $a$ in $Q$ that can synchronize with $\overline{a}$ in order to execute. The only restricted actions created by \autoref{alg:ccs-net} are generated for transitions with two ingoing edges, namely for $t \in T_2$. Hence, $a = s_t$ and $\overline{a} = \overline{s_t}$. Similar to \emph{Case 2-1}, it is possible to conclude that there must be two distinct places $p^*, p^{**} \in P$ that have an edge to $t$, formally $(p^*, t), (p^{**}, t) \in F$. Therefore, one of the choices in $\mathcal{D}(X_{p^*})$ must contain one of the replacements for $T_t$ while $\mathcal{D}(X_{p^{**}})$ contains the replacement that is not a choice in $\mathcal{D}(X_{p^*})$. $s_t$ and $\overline{s_t}$ are not present anywhere else and since they can synchronize, then there must be at least one $X_{p^*}$ and at least one $X_{p^{**}}$. Thus, $M(p^*) \geq 1$ and $M(p^{**}) \geq 1$ such that $t$ is enabled in $N$.

    The rest of the proof for this case is the same as \emph{Case 1-2}. Hence, \emph{Case 2} shows that $LTS(N, M_0)$ strongly simulates $LTS(Q_0, \mathcal{D})$ using the relation $\mathcal{R}^{-1} \mathrel{\coloneqq} \left\{(Q, M) \;\middle|\; (M, Q) \in \mathcal{R}\right\}$. This concludes the proof of $LTS(N, M_0) \bisimilar LTS(Q_0, \mathcal{D})$.

    \paragraph{Translation time} Create an array $R$ of size $|T|$ where each element $R[t]$ is an array of size two with references to placeholders ($Y_t$). While iterating over the places (line 2-4 in \autoref{alg:ccs-net}), populate $R$ with references to the placeholders. When iterating over the transitions (line 5-13 in \autoref{alg:ccs-net}), it is possible to lookup where the replacement should go in constant time.

    It takes $\O(|T|)$ time to create $R$. Defining all $\mathcal{D}(X_p)$ takes $\O(|P| + |F|)$ because there are $|P|$ definitions with a total of at most $|F|$ placeholders (and it takes constant time to set a reference in $R$). Replacing all placeholders for transitions with one ingoing edge takes $\O(|T| + |F|)$ time since there are at most $|T|$ actions ($\sigma(t)$) and at most $|F|$ process names. Replacing all placeholders for transitions with two ingoing edge takes $\O(|T| + |F|)$ time because at most $|T|$ new actions are created, at most $2 |T|$ replacements are performed in constant time and it takes at most $|F|$ time to create the replacements. Defining $Q_0$ takes $\O(|T| + \sum_{p \in P} M_0(p))$ (or $\O(|N|)$) time. In total, the time usage is $\O(|P| + |T| + |F| + \sum_{p \in P} M_0(p)) = \O(|N| + \sum_{p \in P} M_0(p))$ (or alternatively $\O(|N|)$).

    \paragraph{Size of $(Q_0, \mathcal{D})$} The size of $Q_0$ is $\O(|T| + \sum_{p \in P} M_0(p))$ since there are at most one new action per transition ($n < |T|$) and there is one process in parallel composition per token in $M_0$. (Alternatively if process exponentiation is considered a constant sized term, then the size of $Q_0$ can be written as $\O(|T| + |P|)$.) The size of all place processes is $\O(|P| + |T| + |F|)$ since there are $|P|$ definitions (that all could be $\mathbf{0}$), each transition action $\sigma(t)$ appears at most once, there are at most $|T|$ new (co-)actions that appear once and there are at most $|F|$ process names in parallel composition. So the total size is $\O(|P| + |T| + |F| + \sum_{p \in P} M_0(p)) = \O(|N| + \sum_{p \in P} M_0(p))$ (or alternatively $\O(|N|)$).
    \qed
\end{proof}
\end{restate}

\section{Proofs for Encoding Free-Choice Workflow Nets into CCS Processes}\label{sec:appendix-free-wf}

\begin{restate}{theorem}{thm:free-wf-stepwise-weak-bisimulation}
\begin{theorem}[Correctness of \autoref{alg:free-wf-stepwise}]\label{thm:appendix-free-wf-stepwise-weak-bisimulation}
    Given a free-choice net $N = (P, T, F, A, \sigma)$, a marking $M_0: P \to \mathbb{N}$, and a transition $t^* \in T$ (with at least two ingoing edges), the result of applying \autoref{alg:free-wf-stepwise} on $N$, $M_0$, and $t^*$ is a Petri net $N'$ and marking $M_0'$ such that $LTS(N, M_0) \wbisimilar LTS(N', M_0')$ and $LTS(N', M_0')$ contains a divergent path iff $LTS(N, M_0)$ contains a divergent path. The transformation time and the increase in size are amortized $\O(1)$.
\end{theorem}
\label{sec:appendix-free-wf-stepwise-weak-bisimulation}
\begin{proof}
    Markings describe states in Petri nets. Define the weak bisimulation relation $\mathcal{R}$ as follows:
    \begin{align*}
        \mathcal{R} \mathrel{\;\coloneqq\;} \left\{\left(M,\ M[p^* \mapsto M(p^*) - i][p^{**} \mapsto M(p^{**}) - i][p^+ \mapsto i]\right) \middle|
        \begin{smallmatrix}
            M \in LTS(N, M_0) \\
            k = \min(M(p^*), M(p^{**})) \\
            0 \leq i \leq k
        \end{smallmatrix}
        \right\}
    \end{align*}
    where $M$ is a reachable marking from the initial marking $M_0$ in the free-choice net $N$ ($M$ is a state in $LTS(N, M_0)$). The intuition behind the relation is: When $i = 0$, it covers the direct extension of a marking in $N$ to one in $N'$ where the new place $p^+$ has no tokens; when $i > 0$, it covers the cases where the new transition $t^+$ has been fired $i$ times without firing $t^*$ afterwards.

    Before proving the weak bisimulation, some properties of the Petri nets and related markings are considered. In the original Petri net $N$, $t^*$ has at least two ingoing edges, namely the ones from $p^*$ and $p^{**}$. Hence, $t^*$ has more than one ingoing edge which means that $p^*$ and $p^{**}$ cannot have more than one outgoing edge each. The relation can be seen as an invariant between firing transitions: For all $(M, M') \in \mathcal{R}$ then $M(p) = M'(p)$ for all $p \in (P' \setminus \{p^*, p^{**}, p^+\})$, $M(p^*) = M'(p^*) + M'(p^+)$ and $M(p^{**}) = M'(p^{**}) + M'(p^+)$. To show the weak bisimulation, it is enough to show that this invariant is preserved after each step.

    \paragraph{Case 0}
    Clearly, the initial marking $M_0$ for $N$ and the initial marking $M_0' = M_0[p^+ \mapsto 0]$ for $N'$ are a pair in the relation $(M_0, M_0') \in \mathcal{R}$ with $i = 0$.

    \paragraph{Case 1}
    Consider two related markings $(M, M') \in \mathcal{R}$ and an arbitrary enabled transition $t \in T$ in $N$ with respect to $M$. Define all places with an edge to $t$ in $N$ as $P_{in} \mathrel{\coloneqq} \left\{p \;\middle|\; (p, t) \in F\right\}$ and all places with an edge from $t$ in $N$ as $P_{out} \mathrel{\coloneqq} \left\{p \;\middle|\; (t, p) \in F\right\}$. Similar for $t$ in $N'$, define $P_{in}' \mathrel{\coloneqq} \left\{p \;\middle|\; (p, t) \in F'\right\}$ and $P_{out}' \mathrel{\coloneqq} \left\{p \;\middle|\; (t, p) \in F'\right\}$.

    \paragraph{Case 1-1}
    First assume that $t \neq t^*$. In this case, $t$ has the same edges in both $N$ and $N'$ which means that $P_{in} = P_{in}'$ and $P_{out} = P_{out}'$. Furthermore, $p^*, p^{**}, p^+ \notin P_{in}$ since $p^*$ and $p^{**}$ do not have an edge to $t$ while $p^+$ is not a part of $N$. Hence, the invariant gives that $M(p) = M'(p)$ for all $p \in P_{in}$ and thus $t$ must also be enabled in $N'$ with respect to $M'$. If $t$ is fired in $N$ then \emph{mimic} the behavior by firing $t$ in $N'$. Let $M_t$ be the marking resulting from firing $t$ in $N$ and $M_t'$ be the marking resulting from firing $t$ in $N'$:
    \begin{align}
        \label{eq:appendix-free-wf-stepwise-t-1}
        M_t(p) &\mathrel{\;\coloneqq\;} \begin{cases}
            M(p) - 1 & \text{if } p \in P_{in} \land p \notin P_{out} \\
            M(p) + 1 & \text{if } p \notin P_{in} \land p \in P_{out} \\
            M(p) & \text{otherwise}
        \end{cases}
        \\
        \label{eq:appendix-free-wf-stepwise-t-2}
        M_t'(p) &\mathrel{\;\coloneqq\;} \begin{cases}
            M'(p) - 1 & \text{if } p \in P_{in}' \land p \notin P_{out}' \\
            M'(p) + 1 & \text{if } p \notin P_{in}' \land p \in P_{out}' \\
            M'(p) & \text{otherwise}
        \end{cases}
    \end{align}
    Now, it should be shown that the invariant holds for $(M_t, M_t')$ by assuming the invariant for $(M, M')$. $M_t(p) = M_t'(p)$ holds for all $p \in (P' \setminus \{p^*, p^{**}, p^+\})$ since all these places have the same change in tokens because $P_{in} = P_{in}'$ and $P_{out} = P_{out}'$. Hence, the definitions give $M(p) + j = M'(p) + j$, where $j \in \{-1, 0, 1\}$, which holds by the assumption $M(p) = M'(p)$.

    $P_{in} = P_{in}'$ and $P_{out} = P_{out}'$ also implies $p^+ \notin P_{in}'$ and $p^+ \notin P_{out}'$ such that $M_t'(p^+) = M'(p^+)$. Furthermore, $p^*, p^{**} \notin P_{in}$ since they only have an edge to $t^*$ in $N$. This means that $p^*$ and $p^{**}$ also have the same change in tokens. Thus, $M_t(p^*) = M_t'(p^*) + M_t'(p^+)$ and $M_t(p^{**}) = M_t'(p^{**}) + M_t'(p^+)$ holds using similar arguments as before regarding the change in tokens. Hence, the invariant is preserved which shows that $(M_t, M_t') \in \mathcal{R}$ for this case.

    \paragraph{Case 1-2}
    Now assume that $t = t^*$ and let the marking resulting from firing $t$ in $N$ be $M_{t^*}$. Note that $P_{in}' = (P_{in} \setminus \{p^*, p^{**}\}) \cup \{p^+\}$ and $P_{out}' = P_{out}$ by the definition of the transformation. This means that $t^*$ is enabled in $N$ with respect to $M$ and that $p^*, p^{**} \in P_{in}$. Since $t^*$ is enabled then $M(p) \geq 1$ for all $p \in P_{in}$ and by the invariant $M(p) = M'(p) \geq 1$ for all $p \in (P_{in} \setminus \{p^*, p^{**}\})$. For $p^*$ and $p^{**}$, the invariant \emph{only} provides that $M(p^*) = M'(p^*) + M'(p^+) \geq 1$ and $M(p^{**}) = M'(p^{**}) + M'(p^+) \geq 1$. There are two cases to consider:

    \paragraph{Case 1-2-1}
    First, assume that $M'(p^+) = 0$ such that $M(p^*) = M'(p^*) \geq 1$ and $M(p^{**}) = M'(p^{**}) \geq 1$. This means that $t^*$ is not enabled in $N'$ with respect to $M'$. However, it is possible to fire $t^+$ in $N'$ with respect to $M'$ which results in the marking $M'' \mathrel{\coloneqq} M'[p^* \mapsto M'(p^*) - 1][p^{**} \mapsto M'(p^{**}) - 1][p^+ \mapsto 1]$. With respect to $M''$ it is possible to fire $t^*$ in $N'$ and the resulting marking is $M_{t^*}'$. The resulting markings can be described as:
    \begin{align}
        \label{eq:appendix-free-wf-stepwise-t-star-1}
        M_{t^*}(p) &\mathrel{\;\coloneqq\;} \begin{cases}
            M(p) - 1 & \text{if } p \in P_{in} \land p \notin P_{out} \\
            M(p) + 1 & \text{if } p \notin P_{in} \land p \in P_{out} \\
            M(p) & \text{otherwise}
        \end{cases}
        \\
        \label{eq:appendix-free-wf-stepwise-t-star-2}
        M_{t^*}'(p) &\mathrel{\;\coloneqq\;} \begin{cases}
            M''(p) - 1 & \text{if } p \in P_{in}' \land p \notin P_{out}' \\
            M''(p) + 1 & \text{if } p \notin P_{in}' \land p \in P_{out}' \\
            M''(p) & \text{otherwise}
        \end{cases}
    \end{align}
    For all $p \in (P' \setminus \{p^*, p^{**}, p^+\})$, it is the case that $M'(p) = M''(p)$. Furthermore, all edges for these places are the same in $N$ and $N'$ which means that $p \in P_{in} \iff p \in P_{in}'$ and $p \in P_{out} \iff p \in P_{out}'$. Thus, $M_{t^*}(p) = M_{t^*}'(p)$ holds for all $p \in (P' \setminus \{p^*, p^{**}, p^+\})$.

    Consider $p^*$ and $p^+$. The transformation gives that $p^+ \in P_{in}'$, $p^* \in P_{in}$ and $p^* \notin P'_{in}$ while $P_{out}' = P_{out}$ implies $p^+ \notin P_{out}'$. Either $p^* \notin P_{out}$ or $p^* \in P_{out}$ which gives two cases:
    \begin{align*}
        \begin{split}
            p^* \notin &P_{out} \land p^* \notin P'_{out}
            \\
            M_{t^*}(p^*) &= M_{t^*}'(p^*) + M_{t^*}'(p^+)
            \\
            M(p^*) - 1 &= M''(p^*) + (M''(p^+) - 1)
            \\
            M(p^*) - 1 &= (M'(p^*) - 1) + (1 - 1)
            \\
            M(p^*) &= M'(p^*) + 0
            \\
            M(p^*) &= M'(p^*) + M'(p^+)
        \end{split}
        &
        \begin{split}
            p^* \in &P_{out} \land p^* \in P'_{out}
            \\
            M_{t^*}(p^*) &= M_{t^*}'(p^*) + M_{t^*}'(p^+)
            \\
            M(p^*) &= (M''(p^*) + 1) + (M''(p^+) - 1)
            \\
            M(p^*) &= ((M'(p^*) - 1) + 1) + (1 - 1)
            \\
            M(p^*) &= M'(p^*) + 0
            \\
            M(p^*) &= M'(p^*) + M'(p^+)
        \end{split}
    \end{align*}
    The first equation is the new invariant we want to show holds. The next two equations applies the definitions of the markings. The second to last equation simplifies the expressions while the last equation uses the assumption $M'(p^+) = 0$ to obtain the old invariant which holds by the relation. This shows that $M_{t^*}(p^*) = M_{t^*}'(p^*) + M_{t^*}'(p^+)$ holds. Similar arguments shows that $M_{t^*}(p^{**}) = M_{t^*}'(p^{**}) + M_{t^*}'(p^+)$ holds. Hence, the invariant is preserved when $M'(p^+) = 0$.

    \paragraph{Case 1-2-2}
    Finally, assume that $M'(p^+) > 0$. In this case $t^*$ is enabled in $N'$ with respect to $M'$ so only $t^*$ is fired in $N'$. Let $M'' \mathrel{\coloneqq} M'$ such that the markings $M_{t^*}$ and $M_{t^*}'$ can be described like in \eqref{eq:appendix-free-wf-stepwise-t-star-1} and \eqref{eq:appendix-free-wf-stepwise-t-star-2}. The invariant holds for all $p \in (P' \setminus \{p^*, p^{**}, p^+\})$ using the same arguments as in \emph{Case 1-2-1}.

    Consider $p^*$ and $p^+$ again with the same properties as in \emph{Case 1-2-1} (except that $M''$ is different). Hence, there is again two cases:
    \begin{align*}
        \begin{split}
            p^* \notin &P_{out} \land p^* \notin P'_{out}
            \\
            M_{t^*}(p^*) &= M_{t^*}'(p^*) + M_{t^*}'(p^+)
            \\
            M(p^*) - 1 &= M''(p^*) + (M''(p^+) - 1)
            \\
            M(p^*) - 1 &= M'(p^*) + (M'(p^+) - 1)
            \\
            M(p^*) &= M'(p^*) + M'(p^+)
        \end{split}
        &
        \begin{split}
            p^* \in &P_{out} \land p^* \in P'_{out}
            \\
            M_{t^*}(p^*) &= M_{t^*}'(p^*) + M_{t^*}'(p^+)
            \\
            M(p^*) &= (M''(p^*) + 1) + (M''(p^+) - 1)
            \\
            M(p^*) &= (M'(p^*) + 1) + (M'(p^+) - 1)
            \\
            M(p^*) &= M'(p^*) + M'(p^+)
        \end{split}
    \end{align*}
    This shows that $M_{t^*}(p^*) = M_{t^*}'(p^*) + M_{t^*}'(p^+)$ holds and using similar arguments for $p^{**}$ then also $M_{t^*}(p^{**}) = M_{t^*}'(p^{**}) + M_{t^*}'(p^+)$ holds. Thus, the invariant is preserved when $M'(p^+) > 0$. Hence, the complete proof of \emph{Case 1} shows that $LTS(N', M_0')$ weakly simulates $LTS(N, M_0)$ using $\mathcal{R}$.

    \paragraph{Case 2}
    Consider two related markings $(M, M') \in \mathcal{R}$ and an arbitrary enabled transition $t \in T'$ in $N'$ with respect to $M'$. Define $P_{in} \mathrel{\coloneqq} \left\{p \;\middle|\; (p, t) \in F\right\}$, $P_{out} \mathrel{\coloneqq} \left\{p \;\middle|\; (t, p) \in F\right\}$, $P_{in}' \mathrel{\coloneqq} \left\{p \;\middle|\; (p, t) \in F'\right\}$ and $P_{out}' \mathrel{\coloneqq} \left\{p \;\middle|\; (t, p) \in F'\right\}$ similar to \emph{Case 1}.

    \paragraph{Case 2-1}
    First, assume that $t \notin \{t^*, t^+\}$. In this case $t$ has the exact same edges in both $N$ and $N'$. Thus, $t$ is also enabled in $N$ with respect to $M$. The rest of the proof for this case is similar to \emph{Case 1-1}.

    \paragraph{Case 2-2}
    Secondly, assume that $t = t^+$. Firing $t^+$ in $N'$, the resulting marking is $M_{t^+}' \mathrel{\coloneqq} M'[p^* \mapsto M'(p^*) - 1][p^{**} \mapsto M'(p^{**}) - 1][p^+ \mapsto M'(p^+) + 1]$. Since $t^+$ is a $\tau$-transition, it is \emph{allowed} to fire no transitions in $N$ to \emph{mimic} this behavior such that the resulting marking is simply $M_{t^+} \mathrel{\coloneqq} M$. Clearly, the invariant is preserved for all $p \in (P' \setminus \{p^*, p^{**}, p^+\})$.

    Now consider the invariant for $p^*$ and $p^+$:
    \begin{align*}
        \begin{split}
            M_{t^+}(p^*) &= M_{t^+}'(p^*) + M_{t^+}'(p^+)
            \\
            M(p^*) &= (M'(p^*) - 1) + (M'(p^+) + 1)
            \\
            M(p^*) &= M'(p^*) + M'(p^+)
        \end{split}
        &
        \begin{split}
            \text{(invariant)}
            \\
            \text{(definition)}
            \\
            \text{(simplification)}
        \end{split}
    \end{align*}
    Thus, the invariant is also preserved for $p^*$ and similar for $p^{**}$. This finishes the proof of $(M_{t^+}, M_{t^+}') \in \mathcal{R}$ for this case.

    \paragraph{Case 2-3}
    Finally, assume that $t = t^*$. Note that $P_{in}' = (P_{in} \setminus \{p^*, p^{**}\}) \cup \{p^+\}$ and $P_{out}' = P_{out}$ by the definition of the transformation. Since $t^*$ is enabled in $N'$ with respect to $M'$ then $M'(p) \geq 1$ for all $p \in P_{in}'$ which means that $M(p) \geq 1$ for all $p \in P_{in} \setminus \{p^*, p^{**}\}$. However, it still needs to be shown that $M(p^*) \geq 1$ and $M(p^{**}) \geq 1$ for $t^*$ to be enabled in $N$. By the invariant, it holds that $M(p^*) = M'(p^*) + M'(p^+)$ and $M(p^{**}) = M'(p^{**}) + M'(p^+)$. Since $p^+ \in P_{in}'$ and $M'(p^+) \geq 1$ then $M(p^*) \geq 1$ and $M(p^{**}) \geq 1$ must be the case. Hence, $t^*$ is also enabled in $N$.

    Let $M_{t^*}$ be the resulting marking by firing $t^*$ in $N$ and $M_{t^*}'$ be the resulting marking by firing $t^*$ in $N'$. They can be described similar to \eqref{eq:appendix-free-wf-stepwise-t-1} and \eqref{eq:appendix-free-wf-stepwise-t-2} (replace $t$ with $t^*$). Clearly, the invariant is preserved for all $p \in (P' \setminus \{p^*, p^{**}, p^+\})$, since $P_{in} \setminus \{p^*, p^{**}, p^+\} = P_{in}' \setminus \{p^*, p^{**}, p^+\}$ and $P_{out} = P_{out}'$.

    Now consider $p^*$ and $p^+$. The transformation gives that $p^+ \in P_{in}'$, $p^+ \notin P_{out}'$, $p^* \in P_{in}$ and $p^* \notin P_{in}'$. Like in \emph{Case 1-2-2}, there are two cases:
    \begin{align*}
        \begin{split}
            p^* \notin &P_{out} \land p^* \notin P'_{out}
            \\
            M_{t^*}(p^*) &= M_{t^*}'(p^*) + M_{t^*}'(p^+)
            \\
            M(p^*) - 1 &= M'(p^*) + (M'(p^+) - 1)
            \\
            M(p^*) &= M'(p^*) + M'(p^+)
        \end{split}
        &
        \begin{split}
            p^* \in &P_{out} \land p^* \in P'_{out}
            \\
            M_{t^*}(p^*) &= M_{t^*}'(p^*) + M_{t^*}'(p^+)
            \\
            M(p^*) &= (M'(p^*) + 1) + (M'(p^+) - 1)
            \\
            M(p^*) &= M'(p^*) + M'(p^+)
        \end{split}
    \end{align*}
    The first equation is the new invariant we want to show, the second equation used the definitions in \eqref{eq:appendix-free-wf-stepwise-t-1} and \eqref{eq:appendix-free-wf-stepwise-t-2} and the last equation is a simplification of the second to obtain the old invariant. Doing the same for $p^{**}$, it has been shown that $M(p^*) = M'(p^*) + M'(p^+)$ and $M(p^{**}) = M'(p^{**}) + M'(p^+)$ holds. This finishes the proof of $(M_{t^*}, M_{t^*}') \in \mathcal{R}$. \emph{Case 2} shows that $LTS(N, M_0)$ weakly simulates $LTS(N', M_0')$ using $\mathcal{R}^{-1} \mathrel{\coloneqq} \left\{(M', M) \;\middle|\; (M, M') \in \mathcal{R}\right\}$. This concludes the proof of showing that $LTS(N, M_0) \wbisimilar LTS(N', M_0')$.

    \paragraph{Divergent paths} In general, weak bisimulation alone does not ensure the absence of new divergent paths~\cite{lts-book}. However, \autoref{alg:free-wf-stepwise} only extends existing paths to $t^*$ by one $\tau$-transition without introducing any new loops. If there are no divergent paths in $LTS(N, M_0)$ then adding one $\tau$-transition to all paths does not introduce any divergent paths since this requires a loop or adding an infinite sequence of $\tau$-transitions. Similar arguments applies the other way around (except a $\tau$-transition is removed instead of added). Hence, no divergent paths are changed, so $LTS(N', M_0')$ contains a divergent path iff $LTS(N, M_0)$ contains a divergent path.

    \paragraph{Transformation time} It is assumed that $P$ and $T$ are represented by doubling arrays and $F$ by the adjacency list representation (using doubling arrays under the hood). Adding a new place and transition takes amortized constant time. Getting the list of edges for $t^*$ takes constant time. Removing two edges from a list can be done in constant time by taking the last two edges. Inserting an edge also takes amortized constant time. A constant number of edges are removed and added. Adding a new entry to the marking $M_0$ also takes amortized constant time assuming it is a doubling array. Hence, there is a constant number of operations that each takes (amortized) constant time and thus the time usage is amortized $\O(1)$.

    \paragraph{Size of $(N', M_0')$} No extra space is needed unless at least one of the doubling arrays need to double its size. This is only needed when it is full and thus using the same arguments for the time usage, the increase in size is constant amortized space.
    \qed
\end{proof}
\end{restate}

\begin{restate}{lemma}{lem:free-wf-stepwise-invariant}
\begin{lemma}[Invariant of \autoref{alg:free-wf-stepwise}]\label{lem:appendix-free-wf-stepwise-invariant}
    If \autoref{alg:free-wf-stepwise} is given a free-choice net $N$, it returns a free-choice net $N'$ that has no additional or changed transitions with no ingoing edges compared to $N$.
\end{lemma}
\label{sec:appendix-free-wf-stepwise-invariant}
\begin{proof}
    $p^*$, $p^{**}$ and $p^+$ all have one outgoing edge after the transformation and thus satisfy unique choice (in \autoref{def:free-choice-net}). $t^+$ has two ingoing edges from the places $p^*$ and $p^{**}$ meaning $t^+$ satisfies unique synchronisation (in \autoref{def:free-choice-net}). The new edge to $t^*$ comes from $p^+$ with one outgoing edge meaning $t^*$ satisfies unique synchronisation. The rest is unchanged such that it is still a free-choice net.
    \qed
\end{proof}
\end{restate}

\begin{restate}{lemma}{lem:free-wf-full-correctness}
\begin{lemma}[\autoref{alg:free-wf-full} output]\label{lem:appendix-free-wf-full-correctness}
    If applied to a free-choice workflow net $(P, T, F, A, \sigma)$ with finite $F$, \autoref{alg:free-wf-full} returns a CCS net.
\end{lemma}
\label{sec:appendix-free-wf-full-correctness}
\begin{proof}
    If \autoref{alg:free-wf-full} never enters the loop, it terminates. Otherwise, the \autoref{alg:free-wf-stepwise} reduces the number of ingoing edges by one for the chosen transition. \autoref{alg:free-wf-stepwise} introduces a new $\tau$-transition with two ingoing edges that does not satisfy the loop condition. Hence, the \autoref{alg:free-wf-stepwise} is only applied to transitions in $T$ that starts with a finite number of edges ($F$). Thus, the loop is exited at some point.

    The input is a free-choice net and \autoref{lem:free-wf-stepwise-invariant} ensures that the current Petri net is always a free-choice net (and further that no transitions are changed to have zero ingoing edges).

    When \autoref{alg:free-wf-full} exits the loop, all $\tau$-transitions have at most two ingoing edges and all other transitions have at most one ingoing edge which is a CCS net.
    \qed
\end{proof}
\end{restate}

\begin{restate}{theorem}{thm:free-wf-full-weak-bisimulation}
\begin{theorem}[Correctness of \autoref{alg:free-wf-full}]\label{thm:appendix-free-wf-full-weak-bisimulation}
    Given a free-choice net $N = (P, T, F, A, \sigma)$ and a marking $M_0: P \to \mathbb{N}$, the result of applying \autoref{alg:free-wf-full} on $N$ and $M_0$ is a Petri net $N'$ and marking $M_0'$ such that $LTS(N, M_0) \wbisimilar LTS(N', M_0')$ and $LTS(N, M_0)$ has a divergent path iff $LTS(N', M_0')$ has.   Both the transformation time and the size of $(N', M')$ are $\O(|N|)$.
\end{theorem}
\label{sec:appendix-free-wf-full-weak-bisimulation}
\begin{proof}
    $LTS(N, M_0) \wbisimilar LTS(N', M_0')$ will be shown by induction on the number of applications of \autoref{alg:free-wf-stepwise}.
    
    The base case is that the original free-choice net and marking are returned by \autoref{alg:free-wf-full} since the loop is never entered. Hence, $N = N'$ and $M_0 = M_0'$ such that $LTS(N, M_0) = LTS(N', M_0')$. These are trivially strongly bisimilar and thus also weakly bisimilar \cite[Exercise 2.61]{lts-book}. 
    
    Assume that \autoref{alg:free-wf-stepwise} has been applied $n$ times to obtain $N'$ and $M_0'$. The induction hypothesis gives that $LTS(N, M_0) \wbisimilar LTS(N', M_0')$. If \autoref{alg:free-wf-stepwise} is not applied on $(N', M_0')$, then $(N', M_0')$ is returned such that $LTS(N, M_0) \wbisimilar LTS(N', M_0')$ as required. Otherwise, \autoref{alg:free-wf-stepwise} is applied on $(N', M_0')$ to obtain $(N'', M_0'')$. \autoref{thm:free-wf-stepwise-weak-bisimulation} gives that $LTS(N', M_0') \wbisimilar LTS(N'', M_0'')$. Transitivity of weak bisimulation \cite[Exercise 2.61 \& 2.62]{lts-book} shows that $LTS(N, M_0) \wbisimilar LTS(N'', M_0'')$ as required.

    \paragraph{Divergent paths} The fact about divergent paths follows from transitivity of (bi-)implication in \autoref{thm:free-wf-stepwise-weak-bisimulation}.

    \paragraph{Transformation time} It is possible to allocate enough space before starting the transformation: count the total number of ingoing edges to transitions, subtract two for each $tau$-transition and one for all other transitions (transitions with no ingoing edges are ignored in all cases). This number is the amount of new places and transitions contained in $N'$ compared to $N$. By doing this, the complexities in \autoref{thm:free-wf-stepwise-weak-bisimulation} are no longer \emph{amortized} such that each transformation uses constant time. In the worst case, all places has an edge to a single non-$\tau$-transition (and no other edges) in which case $|F| - 1$ transformations have to be performed which takes $\O(|F|)$ time. Allocating space enough at the start takes $\O(|N|)$ time which is also the total time usage.

    \paragraph{Size of $(N', M_0')$} \autoref{alg:free-wf-stepwise} adds one place, adds one transition and increases the number of edges by two. Each of these changes add constant space and thus adds $\O(1)$ space in total. In the worst case, all places has an edge to a single non-$\tau$-transition (and no other edges) in which case $|F| - 1$ transformations have to be performed which adds $\O(|F|)$ space to $N'$. $|F| - 1$ places are also added, so $M_0'$ also uses $\O(|F|)$ more space. Hence, the size of $N'$ is $\O(|N|)$.
    \qed
\end{proof}
\end{restate}
\section{Proofs for Encoding any Free-Choice Net into CCS}\label{sec:appendix-free}

\begin{lemma}[\autoref{alg:free} output]\label{lem:appendix-free-finite-net-ccs}
   If applied to a 2-$\tau$-synchronisation net $(P, T, F, A, \sigma)$, \autoref{alg:free} returns valid finite-net CCS.
\end{lemma}
\begin{proof}
    Compared to \autoref{alg:ccs-net} (see \autoref{lem:appendix-ccs-net-finite-net-ccs}), \autoref{alg:free} adds further process names to $Q$ and sequential processes to $\mathcal{D}$ without restrictions which is finite-net CCS.
    \qed
\end{proof}

\begin{restate}{theorem}{thm:free-strong-bisimulation}
\begin{theorem}[Correctness of \autoref{alg:free}]\label{thm:appendix-free-strong-bisimulation}
    Given a 2-$\tau$-synchronisation net $N = (P, T, F, A, \sigma)$ and an initial marking $M_0: P \to \mathbb{N}$, the result of applying \autoref{alg:free} on $N$ and $M_0$ is $(Q_0, \mathcal{D})$ such that $LTS(N, M_0) \bisimilar LTS(Q_0, \mathcal{D})$. The translation time and the size of $(Q_0, \mathcal{D})$ are bound by $\O(|N| + \sum_{p \in P} M_0(p))$. ($Q_0$ is used here instead of $Q$ to avoid confusions in the proof.)
\end{theorem}
\label{sec:appendix-free-strong-bisimulation}
\begin{proof}
    This proof builds on top of \autoref{thm:appendix-ccs-net-strong-bisimulation}. Define the parallel composition of the transition generator processes as:
    \begin{align*}
        Q_{T_0} \mathrel{\;\coloneqq\;} \underbrace{X_{t_1} \mathbin{|} \dots \mathbin{|} X_{t_k}}_{t_i \in \{t_1, \dots, t_k\} = T_0}
    \end{align*}
    Define the (updated) bisimulation relation $\mathcal{R}$ as follows:
    \begin{align*}
        \mathcal{R} \mathrel{\;\coloneqq\;} \left\{\left(M,\ (\nu s_1)\dots(\nu s_n)\left({X_{p_1}^{M(p_1)}} \mathbin{|} \dots \mathbin{|} {X_{p_{|P|}}^{M(p_{|P|})}} \mathbin{|} Q_{T_0}\right)\right) \;\middle|\; M \in LTS(N, M_0)\right\}
    \end{align*}
    where $M$ is a reachable marking from $M_0$ in $N$ (a state in $LTS(N, M_0)$) and $T_0$ is the set of transitions from $T$ with no ingoing edges. For simplicity, the restriction of $s_1, \dots, s_n$ are left out in the rest of the proof.

    \paragraph{Case 0}
    Same as \emph{Case 0} in \autoref{thm:appendix-ccs-net-strong-bisimulation} by adding $Q_{T_0}$ to the parallel composition in $Q_0$.

    \paragraph{Case 1}
    Consider $(M, Q) \in \mathcal{R}$ i.e. a marking $M$ and its related process $Q$. Consider an arbitrary enabled transition $t \in T$ with respect to $M$ in $N$. Define all places with an edge from $t$ in $N$ as $P_{out} \mathrel{\coloneqq} \left\{p \;\middle|\; (t, p) \in F\right\}$ and let $k \mathrel{\coloneqq} |P_{out}|$. Without loss of generality, the places in $P_{out}$ are assumed to be enumerated as the first $k$ of the $|P|$ places i.e. $\{p_1, p_2, \dots, p_k\} = P_{out}$.

    \paragraph{Case 1-1}
    Same as \emph{Case 1-1} in \autoref{thm:appendix-ccs-net-strong-bisimulation} by adding $Q_{T_0}$ to the parallel composition in $Q$ and $Q^\prime$.

    \paragraph{Case 1-2}
    Same as \emph{Case 1-2} in \autoref{thm:appendix-ccs-net-strong-bisimulation} by adding $Q_{T_0}$ to the parallel composition in $Q$ and $Q^\prime$.

    \paragraph{Case 1-3}
    Assume that $t \in T_0$ which implies that $X_t$ is in the parallel composition of $Q$. Hence, $\sigma(t)$ can be executed in $Q$. Firing $t$ in $N$ emits a $\sigma(t)$ and results in $M^\prime$:
    \begin{align}
        \label{eq:appendix-free-strong-bisimulation-t0}
        M^\prime(p) \mathrel{\;\coloneqq\;} \begin{cases}
            M(p) + 1 & \text{if } p \in P_{out} \\
            M(p) & \text{otherwise}
        \end{cases}
    \end{align}
    Executing $\sigma(t)$ in $Q$ results in $Q^\prime$ where one instance of $X_t$ is replaced with $(X_t \mathbin{|} X_{p_1} \mathbin{|} \dots \mathbin{|} X_{p_k})$. This is the same adding one instance of $X_p$ for each $p \in P_{out}$ meaning that $Q^\prime$ is:
    \begin{align*}
        X_{p_1}^{M(p_1) + 1} \mathbin{|} \dots \mathbin{|} X_{p_k}^{M(p_k) + 1} \mathbin{|} X_{p_{k + 1}}^{M(p_{k + 1})} \mathbin{|} \dots \mathbin{|} X_{p_{|P|}}^{M(p_{|P|})} \mathbin{|} \underbrace{X_{t_1}^{1 - 1 + 1} \mathbin{|} X_{t_2} \mathbin{|} \dots \mathbin{|} X_{t_{k^\prime}}}_{t_i \in \{t_1, \dots, t_{k^\prime}\} = T_0}
    \end{align*}
    $p_1, \dots, p_k$ hits the first case while $p_{k + 1}, \dots, p_{|P|}$ hits the second case in \eqref{eq:appendix-free-strong-bisimulation-t0} such that $Q^\prime$ can be rewritten to:
    \begin{align*}
        X_{p_1}^{M^\prime(p_1)} \mathbin{|} \dots \mathbin{|} X_{p_k}^{M^\prime(p_k)} \mathbin{|} X_{p_{k + 1}}^{M^\prime(p_{k + 1})} \mathbin{|} \dots \mathbin{|} X_{p_{|P|}}^{M^\prime(p_{|P|})} \mathbin{|} \underbrace{X_{t_1} \mathbin{|} X_{t_2} \mathbin{|} \dots \mathbin{|} X_{t_{k^\prime}}}_{t_i \in \{t_1, \dots, t_{k^\prime}\} = T_0}
    \end{align*}
    This shows that $(M^\prime, Q^\prime) \in \mathcal{R}$ holds for $t \in T_0$. Hence, the complete proof of \emph{Case 1} shows that $LTS(Q_0, \mathcal{D})$ strongly simulates $LTS(N, M_0)$ using $\mathcal{R}$.

    \paragraph{Case 2}
    Consider $(M, Q) \in \mathcal{R}$ i.e. a process $Q$ and its related marking $M$.

    \paragraph{Case 2-1}
    Consider a non-restricted action $a$ in $Q$ that can be executed. \autoref{alg:free} generates non-restricted actions for transitions with zero or one ingoing edge. There are two cases where $a = \sigma(t)$ in both cases.

    \paragraph{Case 2-1-1}
    If $t \in T_1$, then the proof is the same as \emph{Case 2-1} in \autoref{thm:appendix-ccs-net-strong-bisimulation} by adding $Q_{T_0}$ to the parallel composition in $Q$ and $Q^\prime$.

    \paragraph{Case 2-1-2}
    Otherwise, $t \in T_0$ must be the case. Hence, $t$ has no ingoing edges in $N$ and thus can be fired in $N$. The rest of the proof for this case is the same as \emph{Case 1-3}.

    \paragraph{Case 2-2}
    Same as \emph{Case 2-2} in \autoref{thm:appendix-ccs-net-strong-bisimulation} by adding $Q_{T_0}$ to the parallel composition in $Q$ and $Q^\prime$. Hence, \emph{Case 2} shows that $LTS(N, M_0)$ strongly simulates $LTS(Q_0, \mathcal{D})$ using $\mathcal{R}^{-1} \mathrel{\coloneqq} \left\{(Q, M) \;\middle|\; (M, Q) \in \mathcal{R}\right\}$. This concludes the proof of showing that $LTS(N, M_0) \bisimilar LTS(Q_0, \mathcal{D})$.

    \paragraph{Translation time} For all the transitions with no ingoing edges, there are at most $|T|$ actions, $|T|$ self-referencing process names ($X_t$) and at most $|F|$ other process names in parallel compositions. This uses $\O(|T| + |F|) = \O(|N|)$ time. In addition, at most $|T|$ process names are added to $Q_0$ which adds $\O(|T|)$ time such that the total time is $\O(|N| + \sum_{p \in P} M_0(p))$ (or alternatively $\O(|N|)$).

    \paragraph{Size of $(Q_0, \mathcal{D})$} The size of $Q_0$ is increased by $\O(|T_0|) = \O(|T|)$. The size of the new definitions are similar to the space used by transitions with one ingoing edge except they have one extra process name, but that is still $\O(|T| + |F|)$ space. Hence, the total size of $(Q_0, \mathcal{D})$ is $\O(|N| + \sum_{p \in P} M_0(p))$ (or alternatively $\O(|N|)$).
    \qed
\end{proof}
\end{restate}

\begin{restate}{lemma}{lem:free-correctness}
\begin{lemma}[\autoref{alg:free-wf-full} output]\label{lem:appendix-free-correctness}
    If applied to a free-choice net $(P, T, F, A, \sigma)$ with finite $F$, \autoref{alg:free-wf-full} returns a 2-$\tau$-synchronisation net.
\end{lemma}
\label{sec:appendix-free-correctness}
\begin{proof}
    Similar to the proof of \autoref{lem:free-wf-full-correctness} except the input net is a free-choice net which might have transitions with no ingoing edges which by \autoref{lem:free-wf-stepwise-invariant} means that the result might have them as well. Therefore, the difference is that a 2-$\tau$-synchronisation net is returned instead.
    \qed
\end{proof}
\end{restate}
\section{Proofs for Group-Choice Nets}\label{sec:appendix-group-sec}

\begin{lemma}[Invariant of \autoref{alg:group-stepwise}]\label{lem:appendix-group-stepwise-correctness}
    If \autoref{alg:group-stepwise} is given a group-choice net then it produces a group-choice net.
\end{lemma}
\begin{proof}
    The changes by \autoref{alg:group-stepwise} are considered in steps. $p^+$ gets the place postset as $p^*$ ($p^+\bullet = p^*\bullet$) which satisfies the constraints for group-choice nets since it is the same as $p^*\bullet$. The transition $t^+$ and edge $(t^+, p^+)$ do not affect any place postsets. The place postset for both $p^*$ and $p^{**}$ are changed to $\{t^+\}$. Since $t^+$ is new, it is not contained in any place postset and thus the new place postset for $p^*$ and $p^{**}$ are disjointed from any other place postset. Hence, the constraints for a group-choice net are preserved which concludes the proof.
    \qed
\end{proof}

\begin{theorem}[Correctness of \autoref{alg:group-stepwise}]\label{thm:appendix-group-stepwise-weak-bisimulation}
    Given a group-choice net $N = (P, T, F, A, \sigma)$, a marking $M_0: P \to \mathbb{N}$, and two places $p^*, p^{**} \in P$ (with $p^*\bullet = p^{**}\bullet \neq \emptyset$), the result of applying \autoref{alg:group-stepwise} on $N$, $M_0$, $p^*$ and $p^{**}$ is a Petri net $N'$ and marking $M_0'$  such that $LTS(N, M_0) \wbisimilar LTS(N', M_0')$ and $LTS(N', M_0')$ contains a divergent path iff $LTS(N, M_0)$ contains a divergent path. The transformation time and the increase in size is amortized $\O(1)$.
\end{theorem}
\begin{proof}
    Markings describe states in Petri nets. Define the weak bisimulation relation $\mathcal{R}$ as follows:
    \begin{align*}
        \mathcal{R} \mathrel{\;\coloneqq\;} \left\{\left(M, M[p^* \mapsto M(p^*) - i][p^{**} \mapsto M(p^{**}) - i][p^+ \mapsto i]\right) \middle|
        \begin{smallmatrix}
            M \in LTS(N, M_0) \\
            k = \min(M(p^*), M(p^{**})) \\
            0 \leq i \leq k
        \end{smallmatrix}
        \right\}
    \end{align*}
    where $M$ is a reachable marking from the initial marking $M_0$ in the free-choice net $N$ ($M$ is a state in $LTS(N, M_0)$). The intuition behind the relation is: When $i = 0$, it covers the direct extension of a marking in $N$ to one in $N'$ where the new place $p^+$ has no tokens; when $i > 0$, it covers the cases where the new transition $t^+$ has been fired $i$ times without firing $t^*$ afterwards. $\mathcal{R}$ describes an invariant: For every $(M, M^\prime) \in \mathcal{R}$, $M(p) = M^\prime(p)$ holds for all $p \in (P^\prime \setminus \{p^*, p^{**}, p^+\})$, $M(p^*) = M^\prime(p^*) + M^\prime(p^+)$ and $M(p^{**}) = M^\prime(p^{**}) + M^\prime(p^+)$. It is enough to show that the invariant is preserved to show weak bisimulation.

    \paragraph{Case 0}
    Clearly, the initial marking $M_0$ for $N$ and the initial marking $M_0^\prime = M_0[p^+ \mapsto 0]$ for $N^\prime$ are a pair in the relation $(M_0, M_0^\prime) \in \mathcal{R}$ with $i = 0$.

    \paragraph{Case 1}
    Consider two related markings $(M, M^\prime) \in \mathcal{R}$ and an arbitrary \emph{enabled} transition $t \in T$ in $N$ according to $M$. Define all places with an edge to $t$ in $N$ as $P_{in} \mathrel{\coloneqq} \left\{p \;\middle|\; (p, t) \in F\right\}$ and all places with an edge from $t$ in $N$ as $P_{out} \mathrel{\coloneqq} \left\{p \;\middle|\; (t, p) \in F\right\}$. Similar for $t$ in $N^\prime$, define $P_{in}^\prime \mathrel{\coloneqq} \left\{p \;\middle|\; (p, t) \in F^\prime\right\}$ and $P_{out}^\prime \mathrel{\coloneqq} \left\{p \;\middle|\; (t, p) \in F^\prime\right\}$.

    \paragraph{Case 1-1}
    First assume that $t \notin p^*\bullet$ according to $N$. Hence, $p^*$ and $p^{**}$ do not have an edge to $t$ in $N$. \autoref{alg:group-stepwise} only changes edges for transitions with an edge from $p^*$ or $p^{**}$. Therefore, $t$ has the exact same edges in $N$ and $N^\prime$ meaning that $P_{in} = P_{in}^\prime$ and $P_{out} = P_{out}^\prime$. Thus, $t$ is also enabled in $N^\prime$. Firing $t$ in $N$, the marking $M_t$ is obtained while firing $t$ in $N^\prime$ gives the marking $M_t^\prime$:
    \begin{align}
        \label{eq:appendix-group-stepwise-1}
        M_t(p) &\mathrel{\;\coloneqq\;} \begin{cases}
            M(p) - 1 & \text{if } p \in P_{in} \land p \notin P_{out} \\
            M(p) + 1 & \text{if } p \notin P_{in} \land p \in P_{out} \\
            M(p) & \text{otherwise}
        \end{cases}
        \\
        \label{eq:appendix-group-stepwise-2}
        M_t^\prime(p) &\mathrel{\;\coloneqq\;} \begin{cases}
            M^\prime(p) - 1 & \text{if } p \in P_{in}^\prime \land p \notin P_{out}^\prime \\
            M^\prime(p) + 1 & \text{if } p \notin P_{in}^\prime \land p \in P_{out}^\prime \\
            M^\prime(p) & \text{otherwise}
        \end{cases}
    \end{align}
    It is assumed that the invariant holds for $(M, M^\prime)$. For all places $p \in (P^\prime \setminus \{p^*, p^{**}, p^+\})$, the change in tokens is the same in both $M$ and $M^\prime$ because $P_{in} = P_{in}^\prime$ and $P_{out} = P_{out}^\prime$. Hence, $M(p) + j = M^\prime(p) + j$ holds from the invariant, where $j \in \{-1, 0, 1\}$, which can be rewritten to $M_t(p) = M_t^\prime(p)$ by the definitions.

    For $p^*$, $p^{**}$ and $p^+$, it holds that $p^*, p^{**}, p^+ \notin P_{in}$ and $p^+ \notin P_{out}$. Hence, $M_t^\prime(p^+) = M^\prime(p^+)$ such that similar arguments like above shows $M_t(p^*) = M_t^\prime(p^*) + M_t^\prime(p^+)$ and $M_t(p^{**}) = M_t^\prime(p^{**}) + M_t^\prime(p^+)$. Hence, the invariant is preserved which shows that $(M_t, M_t^\prime) \in \mathcal{R}$.

    \paragraph{Case 1-2}
    Now assume that $t \in p^*\bullet$ according to $N$. In this case, $p^*, p^{**} \in P_{in}$ while $P_{in}^\prime = (P_{in} \setminus \{p^*, p^{**}\}) \cup \{p^+\}$ and $P_{out}^\prime = P_{out}$ according to \autoref{alg:group-stepwise}. Since $t$ is enabled in $N$ then $M(p) \geq 1$ for all $p \in P_{in}$ and by the invariant $M(p) = M^\prime(p) \geq 1$ holds for all $p \in (P_{in} \setminus \{p^*, p^{**}\})$. For $p^*$ and $p^{**}$, the invariant \emph{only} provides $M(p^*) = M^\prime(p^*) + M^\prime(p^+) \geq 1$ and $M(p^{**}) = M^\prime(p^{**}) + M^\prime(p^+) \geq 1$. There are two cases:

    \paragraph{Case 1-2-1}
    First, assume that $M^\prime(p^+) = 0$ such that $M(p^*) = M^\prime(p^*) \geq 1$ and $M(p^{**}) = M^\prime(p^{**}) \geq 1$. Hence, $t$ is not enabled in $N^\prime$ since $p^+$ has no tokens. However, $t^+$ can be fired in $N^\prime$ to obtain $M^{\prime\prime} \mathrel{\coloneqq} M^\prime[p^* \mapsto M^\prime(p^*) - 1][p^{**} \mapsto M^\prime(p^{**}) - 1][p^+ \mapsto 1]$ which enables $t$ in $N^\prime$. The resulting markings are:
    \begin{align}
        M_t(p) &\mathrel{\;\coloneqq\;} \begin{cases}
            M(p) - 1 & \text{if } p \in P_{in} \land p \notin P_{out} \\
            M(p) + 1 & \text{if } p \notin P_{in} \land p \in P_{out} \\
            M(p) & \text{otherwise}
        \end{cases}
        \\
        M_t^\prime(p) &\mathrel{\;\coloneqq\;} \begin{cases}
            M^{\prime\prime}(p) - 1 & \text{if } p \in P_{in}^\prime \land p \notin P_{out}^\prime \\
            M^{\prime\prime}(p) + 1 & \text{if } p \notin P_{in}^\prime \land p \in P_{out}^\prime \\
            M^{\prime\prime}(p) & \text{otherwise}
        \end{cases}
    \end{align}
    For all $p \in (P^\prime \setminus \{p^*, p^{**}, p^+\})$, $M^\prime(p) = M^{\prime\prime}(p)$ such that using the same arguments as in \emph{Case 1-1} shows that $M_t(p) = M_t^\prime(p)$.

    $p^+$ hits the first case above in $N^\prime$ such that $M_t^\prime(p^+) = M^{\prime\prime}(p^+) - 1 = M^\prime(p^+) + 1 - 1 = M^\prime(p^+) = 0$. There are two cases for $p^*$: If $p^* \in P_{out}$ then $p^*$ hits the third case in $N$ and second case in $N^\prime$. Hence, $M_t(p^*) = M_t^\prime(p^*) + M_t^\prime(p^+)$ holds since it is by the definitions equal to $M(p^*) = M^{\prime\prime}(p^*) + 1 + M^\prime(p^+) = M^\prime(p^*) - 1 + 1 + M^\prime(p^+) = M^\prime(p^*) + M^\prime(p^+)$ which hold from the invariant. Otherwise $p^* \notin P_{out}$ such that $p^*$ hits the first case in $N$ and the third case in $N^\prime$. Thus, $M_t(p^*) = M_t^\prime(p^*) + M_t^\prime(p^+)$ holds as the definitions give $M(p^*) - 1 = M^{\prime\prime}(p^*) + M^\prime(p^+) = M^\prime(p^*) - 1 + M^\prime(p^+)$ which holds by the invariant. Similar arguments holds for $p^{**}$ which shows the invariant is preserved when $M^\prime(p^+) = 0$.

    \paragraph{Case 1-2-2}
    Now, assume that $M^\prime(p^+) > 0$ such that $t$ is enabled in $N^\prime$. Hence, $t$ can be fired in both $N$ and $N^\prime$ which results in the markings in \eqref{eq:appendix-group-stepwise-1} and \eqref{eq:appendix-group-stepwise-2}.

    For all $p \in (P^\prime \setminus \{p^*, p^{**}, p^+\})$, the same arguments as in \emph{Case 1-1} shows that $M_t(p) = M_t^\prime(p)$.

    $p^+$ hits the first case in $N^\prime$ such that $M_t^\prime(p^+) = M^\prime(p^+) - 1$. There are two cases for $p^*$: If $p^* \in P_{out}$ then $p^*$ hits the third case in $N$ and the second case in $N^\prime$ meaning that $M_t(p^*) = M_t^\prime(p^*) + M_t^\prime(p^+)$ holds since it can be rewritten to $M(p^*) = M^\prime(p) + 1 + M^\prime(p^+) - 1$ which hold by the invariant. Otherwise $p^* \notin P_{out}$ such that $p^*$ hits the first case in $N$ and the third case in $N^\prime$ meaning $M_t(p^*) = M_t^\prime(p^*) + M_t^\prime(p^+)$ is equal to $M(p^*) - 1 = M^\prime(p^*) + M^\prime(p^+) - 1$ that holds from the invariant. Similar arguments holds for $p^{**}$ which shows the invariant is preserved when $M^\prime(p^+) > 0$.

    \paragraph{Case 2}
    Consider two related markings $(M, M^\prime) \in \mathcal{R}$ and an arbitrary \emph{enabled} transition $t \in T^\prime$ in $N^\prime$ according to $M^\prime$. Define $P_{in}$, $P_{out}$, $P_{in}^\prime$ and $P_{out}^\prime$ like in \emph{Case 1}.

    \paragraph{Case 2-1}
    First, assume that $t \notin p^+\bullet \cup \{t^+\}$ according to $N^\prime$. Then $t$ has the same edges in both $N$ and $N^\prime$. The rest of the proof is the same as \emph{Case 1-1}.

    \paragraph{Case 2-2}
    Next, assume that $t = t^+$. Firing $t^+$ in $N^\prime$ results in the marking $M_t^\prime \mathrel{\coloneqq} M^\prime[p^* \mapsto M^\prime(p^*) - 1][p^{**} \mapsto M^\prime(p^{**}) - 1][p^+ \mapsto M^\prime(p^+) + 1]$. $t^+$ is a $\tau$-transition and does not exist in $N$, so no transitions will be fired in $N$ such that $M_t(p) \mathrel{\coloneqq} M(p)$.

    Clearly, the invariant is preserved for all $p \in (P^\prime \setminus \{p^*, p^{**}, p^+\})$. For $p^*$ and $p^+$, $M_t(p^*) = M_t^\prime(p^*) + M_t^\prime(p^+)$ holds since it by the definitions gives $M(p^*) = M^\prime(p^*) - 1 + M^\prime(p^+) + 1$ which holds by the invariant. The same holds for $p^{**}$ which proves this case.

    \paragraph{Case 2-3}
    Finally, assume that $t \in p^+\bullet$ according to $N^\prime$ (same as $t \in T_{p^*}$ according to $N$). In this case, $P_{in}^\prime = (P_{in} \setminus \{p^*, p^{**}\}) \cup \{p^+\}$ and $P_{out}^\prime = P_{out}$. $t$ is enabled in $N^\prime$ which means that $M^\prime(p) \geq 1$ for all $p \in P_{in}^\prime$ and thus $M(p) \geq 1$ for all $p \in P_{in} \setminus \{p^*, p^{**}\}$. It still needs to be shown that $M(p^*) \geq 1$ and $M(p^{**}) \geq 1$ for $t$ to be enabled in $N$. The invariant gives that $M(p^*) = M^\prime(p^*) + M^\prime(p^+)$ and $M(p^{**}) = M^\prime(p^{**}) + M^\prime(p^+)$. Since $p^+ \in P_{in}^\prime$ and $M^\prime(p^+) \geq 1$ then $M(p^*) \geq 1$ and $M(p^{**}) \geq 1$ must be the case. Hence, $t$ is also enabled in $N$.

    Let the markings obtained by firing $t$ in both $N$ and $N^\prime$ be like in \eqref{eq:appendix-group-stepwise-1} and \eqref{eq:appendix-group-stepwise-2}. The rest of the proof for this case is the same as \emph{Case 1-2-2}. This concludes the proof of showing that $LTS(N, M_0) \wbisimilar LTS(N^\prime, M_0^\prime)$.

    \paragraph{Divergent paths} In general, weak bisimulation alone does not ensure the absence of new divergent paths~\cite{lts-book}. However, \autoref{alg:free-wf-stepwise} only extends existing paths to $t^*$ by one $\tau$-transition without introducing any new loops. If there are no divergent paths in $LTS(N, M_0)$ then adding one $\tau$-transition to all paths does not introduce any divergent paths since this requires a loop or adding an infinite sequence of $\tau$-transitions. Similar arguments applies the other way around (except a $\tau$-transition is removed instead of added). Hence, no divergent paths are changed meaning $LTS(N', M_0')$ contains a divergent path iff $LTS(N, M_0)$ contains a divergent path.

    \paragraph{Transformation time} It is assumed that $F$ is represented as an adjacency list and everything uses doubling arrays (under the hood). Creating $p^+$ and $t^+$ and extending all structures to fit them takes amortized $\O(1)$ time. Adding the edges to $p^+$ can be done in $\O(1)$ time by simply moving (the reference to) the list of edges from $p^*$ to $p^+$. Hence, replacing all edges from $p^*$ and $p^{**}$ can be done by removing the whole list of edges and add a new one with one element which takes $\O(1)$ time. Adding a new entry to the marking $M_0$ also takes amortized $\O(1)$ time. In total, amortized $\O(1)$ time.

    \paragraph{Size of $(N', M_0')$} The new edges to $t^+$ uses $\O(1)$ extra space. Everything else only increases the size if at least one of the doubling arrays need to double its size. This is only needed when an array is full. Thus, the size is increased by amortized $\O(1)$ space.
    \qed
\end{proof}

\begin{lemma}[\autoref{alg:group-full} output]\label{lem:appendix-group-full-correctness}
    If applied to a group-choice net $(P, T, F, A, \sigma)$ with finite $F$, \autoref{alg:group-full} returns a 2-$\tau$-synchronisation net.
\end{lemma}
\begin{proof}
    If \autoref{alg:group-full} never enters the loop, it terminates. Otherwise, the \autoref{alg:group-stepwise} reduces the number of ingoing edges by one for each transition in $p^*\bullet$. \autoref{alg:group-stepwise} introduces a new $\tau$-transition with two ingoing edges that does not satisfy the loop condition. Hence, the \autoref{alg:group-stepwise} is only applied to transitions in $T$ that starts with a finite number of edges ($F$). Thus, the loop is exited at some point.

    The input is a group-choice net and \autoref{lem:appendix-group-stepwise-correctness} ensures that the current Petri net is always a group-choice net.

    When \autoref{alg:group-full} exits the loop, all $\tau$-transitions have at most two ingoing edges and all other transitions have at most one ingoing edge which is a 2-$\tau$-synchronisation net.
    \qed
\end{proof}

\begin{theorem}[Correctness of \autoref{alg:group-full}]\label{thm:appendix-group-full-weak-bisimulation}
    Given a group-choice net $N = (P, T, F, A, \sigma)$ and a marking $M_0: P \to \mathbb{N}$, the result of applying \autoref{alg:group-full} on $N$ and $M_0$ is a Petri net $N'$ and marking $M_0'$ such that $LTS(N, M_0) \wbisimilar LTS(N', M_0')$ and $LTS(N, M_0)$ has a divergent path iff $LTS(N', M_0')$ has. Both the transformation time and the size of $(N', M')$ are $\O(|N|)$.
\end{theorem}
\begin{proof}
    $LTS(N, M_0) \wbisimilar LTS(N', M_0')$ will be shown by induction on the number of applications of \autoref{alg:group-stepwise}.
    
    The base case is that the original group-choice net and marking are returned by \autoref{alg:group-full} since the loop is never entered. Hence, $N = N'$ and $M_0 = M_0'$ such that $LTS(N, M_0) = LTS(N', M_0')$. These are trivially strongly bisimilar and thus also weakly bisimilar \cite[Exercise 2.61]{lts-book}. 
    
    Assume that \autoref{alg:group-stepwise} has been applied $n$ times to obtain $N'$ and $M_0'$. The induction hypothesis gives that $LTS(N, M_0) \wbisimilar LTS(N', M_0')$. If \autoref{alg:group-stepwise} is not applied on $(N', M_0')$, then $(N', M_0')$ is returned such that $LTS(N, M_0) \wbisimilar LTS(N', M_0')$ as required. Otherwise, \autoref{alg:group-stepwise} is applied on $(N', M_0')$ to obtain $(N'', M_0'')$. \autoref{thm:appendix-group-stepwise-weak-bisimulation} gives that $LTS(N', M_0') \wbisimilar LTS(N'', M_0'')$. Transitivity of weak bisimulation \cite[Exercise 2.61 \& 2.62]{lts-book} shows that $LTS(N, M_0) \wbisimilar LTS(N'', M_0'')$ as required.

    \paragraph{Divergent paths} The fact about divergent paths follows from transitivity of (bi-)implication in \autoref{thm:appendix-group-stepwise-weak-bisimulation}.

    \paragraph{Transformation time} Using the same arguments as in \autoref{thm:free-wf-full-weak-bisimulation} (free-choice nets are also group-choice nets), the time usage is $\O(|N|)$.

    \paragraph{Size of $(N', M_0')$} Using the same arguments as in \autoref{thm:free-wf-full-weak-bisimulation}, the size of $N'$ is $\O(|N|)$.
    \qed
\end{proof}

\begin{restate}{theorem}{thm:group}
\begin{theorem}[Correctness of group-choice net to CCS encoding]\label{thm:appendix-group}
    A group-choice net $N = (P, T, F, A, \sigma)$ and an initial marking $M_0: P \to \mathbb{N}$ can be encoded into a weakly bisimilar CCS process $Q$ with defining equations $\mathcal{D}$ s.t. $LTS(N, M_0)$ has a divergent path iff $LTS(Q, \mathcal{D})$ has. The encoding time and the size of $(Q, \mathcal{D})$ are $\O(|N| + \sum_{p \in P} M_0(p))$.
\end{theorem}
\label{sec:appendix-group}
\begin{proof}
    By \autoref{lem:appendix-group-stepwise-correctness}+\ref{lem:appendix-group-full-correctness}, \autoref{thm:appendix-group-stepwise-weak-bisimulation}+\ref{thm:appendix-group-full-weak-bisimulation} and transitivity of bisimulation~\cite{lts-book}.
    \qed
\end{proof}
\end{restate}

\end{document}